\documentclass[12pt,fleqn,leqno,usenames,dvipsnames]{article}
\usepackage[margin=1.25 in]{geometry}
\usepackage{amsmath,amssymb,amsthm}
\usepackage{palatino,pxfonts,wasysym}
\usepackage{indentfirst,geometry,setspace,fancyhdr,sectsty}
\usepackage{relsize,url}
\usepackage{multirow}
\usepackage{dirtytalk}
\usepackage{multirow,array}
\usepackage{dsfont}
\usepackage{textcomp}
\usepackage{tabularx}
\usepackage{listings}

\usepackage{stackengine}
\usepackage[T1]{fontenc}

\usepackage{natbib}
\newcommand{\mybibliography}{\renewcommand{\baselinestretch}{1}\small\normalsize\bibliography{mybibdata}}

\usepackage[unicode=true,pdfusetitle,
 bookmarks=true,bookmarksnumbered=false,bookmarksopen=false,
 breaklinks=false,pdfborder={0 0 1},backref=false,colorlinks=true]
 {hyperref}
\hypersetup{pdfborderstyle=,citecolor=blue,linkcolor=blue,urlcolor=blue}

\usepackage{graphicx}
\usepackage[justification=centering,font=footnotesize]{caption}
\usepackage[labelformat=simple]{subfig}

\usepackage[capitalize,noabbrev,nameinlink]{cleveref}
\numberwithin{equation}{section}
\newcommand{\comment}[1]{}

\fancypagestyle{plain}{\cfoot{}}

\allsectionsfont{\mdseries}

\makeatletter
\def\@seccntformat#1{\@ifundefined{#1@cntformat}{\csname the#1\endcsname\hspace{.5em}}{\csname #1@cntformat\endcsname}}
\def\section@cntformat{\thesection.\hspace{.5em}}
\makeatother

\allowdisplaybreaks

\newcounter{mylistcounter}

\makeatletter
\newcommand\lilhead{\@startsection{paragraph}{4}{0em}{.1in}{.07in}{\bfseries\itshape}}
\makeatother

\newtheorem{theorem}{Theorem}
\newtheorem{proposition}{Proposition}
\newtheorem{corollary}{Corollary}

\newtheorem{claim}{Claim}
\newtheorem{fact}{Fact}

\makeatletter
\newtheorem*{rep@theorem}{\rep@title}
\newcommand{\newreptheorem}[2]{%
\newenvironment{rep#1}[1]{%
 \def\rep@title{#2 \ref{##1}}%
 \begin{rep@theorem}}%
 {\end{rep@theorem}}}
\makeatother

\newreptheorem{theorem}{Theorem}
\newreptheorem{proposition}{Proposition}
\newreptheorem{corollary}{Corollary}
\newreptheorem{lemma}{Lemma}
\newreptheorem{observation}{Observation}

\newtheorem{definition}{Definition}

\newtheoremstyle{examplestyle}{}{}{}{}{\itshape}{}{.5em}{\thmname{#1}\thmnumber{ #2.}\thmnote{ #3.}}
\theoremstyle{examplestyle}

\newtheorem{remark}{Remark}
\newtheorem{factproof}{Proof of Fact}

\DeclareFontFamily{OT1}{pzc}{}
\DeclareFontShape{OT1}{pzc}{m}{it}{<-> s * [1.200] pzcmi7t}{}
\DeclareMathAlphabet{\mathscr}{OT1}{pzc}{m}{it}

\DeclareMathOperator*{\argmax}{argmax}

\renewcommand{\ast}{{\mathlarger *}}

\usepackage[textwidth=30mm]{todonotes}
\urldef{\hvte}\url{yuval.heller@biu.ac.il}
\urldef{\hvtw}\url{https://sites.google.com/site/yuval26/}
\urldef{\sae}\url{arigapudi@ssc.wisc.edu}
\urldef{\saw}\url{sites.google.com/view/srinivasarigapudi/home}
\urldef{\whse}\url{igal.milchtaich@biu.ac.il}
\urldef{\whsw}\url{https://faculty.biu.ac.il/~milchti/}

\definecolor{green}{RGB}{0, 130, 0}
\title{
Instability of Defection in the Prisoner's Dilemma Under Best Experienced
Payoff Dynamics
\thanks{We wish to thank Luis R. Izquierdo, Segismundo S. Izquierdo, Ron Peretz, William Sandholm 
and two anonymous referees for various helpful comments and suggestions. YH and SA gratefully acknowledge the financial support of the European Research Council (starting grant 677057). SA gratefully acknowledges the financial support of the Sandwich research fellowship of Bar-Ilan University and the Israeli Council of Higher Education.}}
\author{Srinivas Arigapudi\thanks{
Department of Economics,
University of Wisconsin, 
1180 Observatory Drive, 
Madison, WI  53706.
e-mail: \sae; 
website: \saw.}
,\hspace{1pt} 
Yuval Heller\thanks{
Department of Economics,
Bar-Ilan University, 
Ramat Gan 5290002, Israel.
e-mail: \hvte; 
website: \hvtw.}
,\hspace{1pt} 
Igal Milchtaich\thanks{
Department of Economics,
Bar-Ilan University, 
Ramat Gan 5290002, Israel.
e-mail: \whse; 
website: \whsw.}
} 
\date{\today}
\begin{document}
\maketitle
\noindent Final pre-print of a manuscript accepted for publication in the \emph{Journal of Economic Theory}.
\begin{abstract}

We study population dynamics under which each revising agent tests each action $k$ times, with each trial being against a newly drawn opponent, and chooses the action whose mean payoff was highest {during the testing phase}. When $k=1,$ defection is globally stable in the prisoner's dilemma. By contrast, when $k>1$ we show that, 
{if the gains from defection are not too large,}
 there exists a globally stable state in which agents cooperate with probability between $28\%$ and $50\%$. Next, we characterize stability of strict equilibria in general games. 
Our results demonstrate that the empirically plausible case of $k>1$ can yield qualitatively different predictions than the case $k=1$ commonly studied in the literature.\\

\textbf{Keywords: }learning, cooperation, best experienced payoff dynamics, sampling equilibrium, evolutionary stability.  \textbf{JEL codes: }C72, C73.
\end{abstract}

\section{Introduction}
The standard approach in game theory assumes that players
form beliefs about the various uncertainties they face and then best respond to these beliefs. In equilibrium, the beliefs will be correct and the players will play a Nash equilibrium. However, in some 
economic environments where the players have limited information about the strategic situation, Nash equilibrium prediction is hard to justify. Consider the following example from \cite{osborne1998games}. 
You are new to town and are planning your route to work. How do you decide which road to take? You know that other people use the roads, but have no idea which road is most congested. One plausible procedure is to try each route several times and then permanently adopt the one that was (on average) best. The outcome of this procedure is stochastic: you may sample the route that is in fact the best on a day when a baseball game congests it. Once you select your route, you become part of the environment that determines other drivers' choices.

This procedure is formalized as follows. Consider agents in a large population who are randomly matched to play a symmetric $n$-player game with a finite set of actions. Agents  occasionally revise their action 
(which can also be interpreted as agents occasionally leaving the population and being replaced by new agents
who base their behavior on the sampling procedure, as in the motivating example above).
Each revising agent samples each feasible action $k$ times and chooses the action that yields the highest average payoff (applying some tie-breaking rule).

This procedure induces a dynamic process according to which the distribution of actions in the population evolves (best experienced payoff dynamics: \citealp{sethi2000stability, sandholm2019best}). An $S(k)$ equilibrium {$\alpha^\ast$ is a rest point of the above dynamics.  The equilibrium is locally stable if any distribution of actions in the population that is sufficiently close to $\alpha^\ast$ converges to  $\alpha^\ast$, and 
globally stable if any distribution of actions in the population with support that includes all actions converges to $\alpha^\ast.$}

The existing literature on payoff sampling equilibria (as surveyed below) has mainly focused on $S(1)$ equilibria, due to their tractability. It seems plausible that real-life behavior would rely on sampling each action more than once. A key insight of our analysis is that sampling actions several times might lead to qualitatively different results than sampling each action only once. In particular, in the prisoner's dilemma game,  $S(1)$ dynamics yield the Nash equilibrium behavior, while $S(k)$ dynamics (for $k>1$) {may} induce substantial cooperation. 

Recall that  each player in the prisoner's dilemma game has two actions, cooperation  $c$ and defection $d$, and the payoffs are as in Table \ref{tab:Payoff-Matrix-of}, where $g, l > 0.$ \cite{sethi2000stability} has shown that defection is the unique $S(1)$ globally stable  equilibrium. 
By contrast, our first main result (Theorem \ref{thm:mainresult}) shows  that for any $k\geq2$, a game for which the gains from defection are not too large (specifically, $g,l<\frac{1}{k-1}$) admits a globally stable state in which the rate of cooperation is between $28\%$ and $50\%$.

\begin{table}[h]
\centering{}
\begin{tabular}{cc|cc}
  &  & \textcolor{red}{\emph{c }}\textcolor{red}{{} } & \textcolor{red}{\emph{d}}\tabularnewline
\cline{2-4} \cline{3-4} \cline{4-4}
\multirow{2}{*} & \textcolor{blue}{\emph{c}}\textcolor{blue}{{} } & \textcolor{blue}{1}~~,~~\textcolor{red}{1}  & \textcolor{blue}{~~-$l$}~,~\textcolor{red}{1+$g$} \tabularnewline
 & \textcolor{blue}{\emph{d }}\textcolor{blue}{{} } & \textcolor{blue}{1+$g$}~,~\textcolor{red}{-$l$~~~}  & \textcolor{blue}{0}~~,~~\textcolor{red}{0} \tabularnewline
\end{tabular}
\begin{centering}
\caption{
Prisoner's Dilemma Payoff Matrix (${\color{blue}g},{\color{red}l}>0$) \label{tab:Payoff-Matrix-of}}
\par\end{centering}
\end{table}

{Our remaining results characterize the local stability of strict equilibria for $k \ge 2$. Proposition \ref{prop:Sk-strict-PD} shows that defection in the prisoner's dilemma game is locally stable iff $l>\frac{1}{k-1}$.
Theorem \ref{thm:unstable sn} extends the analysis to general symmetric games.}
 It presents a simple necessary and sufficient condition for {a strict symmetric equilibrium action $a^\ast$} to be $S(k)$ locally stable ({improving on} the conditions presented in \citealp{sethi2000stability, sandholm2020stability}). Roughly speaking, the condition is that in any set of actions $A^\prime$ that does not include $a^\ast$ there is an action that never yields the highest payoff when the corresponding sample includes a single occurrence of an action in $A^\prime$ and all the other sampled actions are $a^\ast$. 
 {
 Theorem \ref{thm:asymmetric unstable sn} extends the characterization of local stability of strict equilibria to general asymmetric games.}

\textbf{Outline}: {In the remaining parts of the Introduction we review the related literature, and compare our  predictions with the experimental findings.} In Section \ref{sec:model}, we introduce our model and {the} solution concept. We analyze the prisoner's dilemma in Section \ref{sec:3}, and characterize the stability of strict equilibria in general symmetric games in Section \ref{sec:symmetric games}. 
{An extension of the analysis to asymmetric games is presented in Section \ref{sec:asymmetric games}.}

\subsection{{Related Experimental Literature and Testable Predictions}\label{subsec:predictions}}  
{
The typical length of a lab experiment, as well as the subjects' cognitive costs, are likely to limit the sample sizes used
by subjects to test the various  actions  to small values such as $k=2$ or $k=3$ (because larger samples induce 
too-high costs of non-optimal play during the sampling stage, and require larger cognitive effort to analyze). Proposition \ref{pro:k-2-3} shows that for these small sample sizes of $k=2$ or $3$, the $S(k)$ dynamics admit a unique globally stable equilibrium, which depends 
on the parameter $l$. Specifically, everyone defecting is the globally stable equilibrium if $l>\frac{1}{k-1}$, while there is a substantial rate of cooperation between $24\%$ and $33\%$  if $l<\frac{1}{k-1}.$}

{The predictions of our model  match quite well the empirical findings of the meta-study of \citet{mengel2018risk}  concerning the behavior of subjects playing the one-shot prisoner's dilemma.  \citet[Tables A.3, B.5]{mengel2018risk} summarizes 29 sessions of lab experiments of that game in a ``stranger'' (one-shot) setting from 16 papers (with various values of $g$ and $l$, {both} with median 1; {the distribution of values is presented in Appendix \ref{sec-experiments-g-l}}). The average rate of  cooperation in these experiments is $37\%$.
Our predictions are also broadly consistent with the experimentally observed comparative statics with respect to $g$ and $l$, which is that the rate of cooperation is decreasing in $l$ but is independent of $g$ (see \citealp[Table B.5]{mengel2018risk}, where $l$ is called $RISK^{Norm}$ and $g$ is called $TEMPT^{Norm}$).}\footnote{{Theorem 
 \ref{thm:mainresult} shows that for $g$ that is not too large (specifically, $g<\frac{1}{k-1}$), the minimal rate of cooperation in the globally stable state is $28\%$ (when $l<\frac{1}{k-1}$). Proposition \ref{pro:k-2-3} allows arbitrary large $g$, which has the modest impact of slightly decreasing the minimal globally stable rate of cooperation to 24\%.}}

The empirically observed average cooperation rate of $37\%$ can also be explained by other theories. Specifically, it can be explained by agents making errors when the payoff differences are small (quantal response equilibrium, \citealp{mckelvey1995quantal}), or by agents caring about the payoffs of cooperative opponents (inequality aversion $\grave{\textrm{a}}$ la \citealp{fehr1999theory}, and reciprocity $\grave{\textrm{a}}$ la \citealp{rabin1993incorporating}). Our model has two advantages in comparison with these alternative models. First, our model is  parameter free (for a fixed $k$), while the existing models may require tuning their parameters to fit the experimental data (such as the parameter describing the agents' error rates in a quantal response equilibrium). 

Second, the predictions of the existing models are arguably less compatible with the above-mentioned experimentally observed comparative statics. Quantal response equilibrium predicts that the cooperation rate decreases in both parameters. The other models predict that the cooperation rate decreases in $g$ (because an increasing $g$ increases the material payoff from defecting, while it does not change the payoff of a cooperative opponent), and their prediction with respect to $l$ is ambiguous, because increasing $l$ has two opposing effects: increasing the material gain from defection against a defecting opponent but decreasing the payoff of a cooperating opponent.

Our predictions might have an even better fit with experiments in which subjects have only partial information about the payoff matrix (a setting that might be relevant to many real-life interactions),
such as a ``black box'' setting in which players  do not know the game's structure and observe only their realized payoffs (see, e.g., \citealp{nax2015directional, nax2016learning,burton2017social}).

\subsection{Related {Theoretical} Literature}\label{sec:related literature}
The payoff sampling dynamics approach employed in this paper was pioneered by \cite{osborne1998games} and \cite{sethi2000stability}. The approach has been used in a variety of applications, including 
bounded-rationality models in industrial organization (\citealp{spiegler2006competition,spiegler2006market}), 
coordination games (\citealp{ramsza2005stability}), trust and delegation of control (\citealp{rowthorn2008procedural}), market entry (\citealp{chmura2011minority}), ultimatum games (\citealp{mikekisz2013sampling}), common-pool resources (\citealp{cardenas2015stable}), contributions to public goods (\citealp{mantilla2018efficiency}), and finitely repeated games (\citealp{Raj}).

Most of these papers mainly focus on $S(1)$ dynamics, in  which each action is only sampled once.\footnote{{ See also the variant of the $S(1)$ dynamics presented in \cite{rustichini2003equilibria}, according to which after an initial phase of sampling each action once, each player in each round chooses the action that has yielded the highest average payoff so far.}}
One exception is \cite{sandholm2019best}, which analyzes the stable $S(k)$ equilibrium in a centipede game and 
shows that it involves cooperative behavior even when the number of trials $k$ of each action is large. Another is \cite{sandholm2020stability}, which presents general stability and instability criteria of $S(k)$ equilibria in general classes of games, thus providing a unified way of deriving many of the specific results the above papers derive, as well as several new results.

A related, alternative approach is \emph{action sampling dynamics} (or sample best-response dynamics), according to which each revising agent obtains a small random sample of other players' actions, and chooses the action that is a best reply to that sample (see, e.g., \citealp{sandholm2001almost, kosfeld2002myopic, kreindler2013fast, oyama2015sampling, heller2018social, salant2020statistical}). The action sampling approach is a plausible heuristic when the players know the payoff matrix and are capable of strategic thinking but do not know the exact distribution of actions in the population.

\section{Model}\label{sec:model}
We consider a unit-mass continuum of agents who are randomly matched to play a symmetric {$n$-player} game $G=\{A,u\}$, where $A = \{a_1, a_2, \dots, a_m\}$ is the (finite) set of actions and $u: A^n \rightarrow \mathbb{R}$ is the payoff function,
{which is invariant to permutations of its second through $n$-th arguments. An agent taking action $a^1$ against  opponents playing $a^2, \dots,a^n$, in any order, receives payoff $u(a^1,a^2, \dots,a^n)$.}

Aggregate behavior in the population is described by a \textit{population state} $\alpha$ lying in the unit simplex $\Delta \equiv \{\alpha =(\alpha_{a_i})_{i=1}^m \in \mathbb{R}^m_{+} \ | \ \sum_{i=1}^m \alpha_{a_i} = 1\},$ with $\alpha_{a_i}$ representing the fraction of agents in the population using action $a_i .$ The standard basis vector $e_a \in \Delta$ represents the pure, or monomorphic, state in which all agents play action $a.$ 
Where no confusion is likely, we identify the action with the monomorphic state, denoting them both by $a$.
The set of \emph{interior population states}, in which all actions are used by a positive mass of agents, is $Int(\Delta)\equiv \Delta \cap \mathbb{R}^m_{++}$.

A sampling procedure involves the testing of the different actions against randomly drawn opponents, as explained next. Agents occasionally receive opportunities to switch actions 
(equivalently, this can be thought of as agents dying and being replaced by new agents). These opportunities do not depend on the currently used actions. That is, 
when the population state is $\alpha(t)$, the proportion of agents originally using an action $a$ out of the 
agents who revise between time $t$ and $t+dt$ is equal to their proportion in the population $\alpha_{a}(t)$.  

When an agent receives a revision opportunity, he tries each of the feasible actions 
$k$ times, using it each time against a newly drawn opponent from the population. Thus, the probability that the opponent's action is any $a \in A$ is $\alpha_{a}(t)$. The agent then chooses the action that yielded the highest mean payoff in these trials, employing some tie-breaking rule if more than one action yields the highest mean payoff. All of our results hold for any tie-breaking rule. 
Denote the probability that the chosen action is $a$ by $w_{a,k}(\alpha(t))$.

As a result of the revision procedure described above, the expected change in the number of agents using an action $a$
during an infinitesimal time interval of duration $dt$ is
\begin{equation}\label{eq:expected change}
w_{a,k}(\alpha(t))dt - \alpha_a(t)dt.
\end{equation}The first term in \eqref{eq:expected change} is an inflow term, representing the expected number of revising agents who switch to action $a$, while the second term is an outflow term, representing the expected number of revising agents who are currently playing that action. In the limit $dt\to 0,$ the rate of change of the fraction of agents using each action is given {in vector notation by}  
\begin{equation}\label{eq:BEP}
    \dot{\alpha} = w_{k}(\alpha(t)) - \alpha(t),
\end{equation}{where $w_k$ is a vector whose $a$-th component is $w_{a,k}$.} The system of differential equations \eqref{eq:BEP} is called the $k$-\textit{payoff sampling dynamic}. Its rest points are called $S(k)$ equilibria.

\begin{definition}[\citealp{osborne1998games}]\label{def:rest-point}
\emph{A population state $\alpha^\ast\in\Delta$ is an \emph{$S(k)$ equilibrium} if  $ w_{k}(\alpha^\ast) = \alpha^\ast$.}
\end{definition}

An equilibrium is (locally) asymptotically stable if a population beginning near it remains close and eventually converges to the equilibrium, and it is (almost) globally asymptotically stable if the population converges to it from any initial interior state. 

\begin{definition}\label{def:local-stability}
\emph{An $S(k)$ equilibrium $\alpha^{*}$ is \emph{
asymptotically stable} if:}
\emph{
\begin{enumerate}
\item
(\emph{Lyapunov stability}) for every neighborhood $U$ of $\alpha^{*}$ in $\Delta$ there is a neighborhood $V \subset U$ of $\alpha^{*}$ such that if $\alpha(0) \in V $, then $\alpha(t) \in U $ for all $t >0$; and
\item
there is some neighborhood $U$ of $\alpha^{*}$ in $\Delta $ such that all trajectories initially in $U$ converge to $\alpha^{*}$; that is, 
$\alpha\left(0\right)\in U$ implies $\lim_{t\rightarrow\infty}\alpha\left(t\right)=\alpha^{*}.$
\end{enumerate}
}
\end{definition}

\begin{definition}\label{def:global-stability}
\emph{An $S(k)$ equilibrium $\alpha^{*}$ is \emph{
globally asymptotically stable} if 
all interior trajectories converge to $\alpha^{*}$; that is, 
$\alpha\left(0\right)\in Int(\Delta)$ implies $\lim_{t\rightarrow\infty}\alpha\left(t\right)=\alpha^{*}.$}
\end{definition}

\section{
{The Prisoner's Dilemma}}\label{sec:3}
This section focuses on the  {(two-player)} prisoner's dilemma game. The set of actions is given by $A = \{c, d\}$, where $c$ is interpreted as cooperation and $d$ as defection. The payoffs are as described in Table \ref{tab:Payoff-Matrix-of}, with $g,l>0$: when both players cooperate they get payoff 1, when they both defect they get 0, and when one player defects and the other cooperates, the defector gets 1+$g$ and the cooperator gets $-l$.

\subsection{{Stability of Defection}}\label{sec:prelim analysis}\label{sec:strict}

{
\citet[Example 5]{sethi2000stability} analyzes the $S(1)$ dynamics and shows that everyone defecting is globally stable.} 

\begin{claim}
[\citealp{sethi2000stability}]
\label{prop:S1}
\label{claim:S1}
{Defection} is $S(1)$ globally asymptotically stable. 
\end{claim}
{The argument behind Claim \ref{claim:S1} is as follows. When an agent samples the action $c$ (henceforth, the \emph{$c$-sample})
her payoff  is higher than when sampling the action $d$ (henceforth, the \emph{$d$-sample})  iff the opponent has cooperated in the $c$-sample and defected in the $d$-sample (which happens with probability $\alpha_c \cdot \alpha_d$). Therefore, the $1$-payoff sampling dynamic is given by
\begin{align*}\dot{\alpha}_{c} & =w_{c,1}(\alpha)-\alpha_{c}=\alpha_{c} \cdot \alpha_{d}-\alpha_{c}=\alpha_{c}(1-\alpha_{c})-\alpha_{c}=-\alpha_{c}^{2} \ ~~~~~(<0 \text{ if } \alpha_c >0).\end{align*}
}
The unique rest point $\alpha^{*}_c =0$ is  the unique $S(1)$ equilibrium, and it is easy to see that it is globally asymptotically stable.

{Our next result shows that for $k \ge 2$ everyone defecting is no longer 
globally asymptotically stable (indeed, it is not even locally asymptotically stable) if $l<\frac{1}{k-1}$.}  

\begin{proposition} \label{prop:Sk-strict-PD}
{Let $k \ge 2$ and assume} that\footnote{The stability of defection in the borderline case of  $l=\frac{1}{k-1}$ depends on the tie-breaking rule, because observing a single $c$ in the $c$-sample and no $c$'s in the $d$-sample produces a tie between the two samples. If one assumes a uniform tie-breaking rule, then action $c$ wins with a probability of $\frac{k}{2} \epsilon - O(\epsilon^2)$, which is greater than $\epsilon$ if and only if $k>2$. Thus, with this rule, defection is stable if $k=2$ and unstable if $k>2$.} $l \neq \frac{1}{k-1}$.
{Defection} is $S(k)$ asymptotically stable if and only if $l > \frac {1}{k-1}$. 
\end{proposition}
{Proposition \ref{prop:Sk-strict-PD} is implied by the results of \citet{sandholm2020stability}, and as we show in Section \ref{subsec-applications}, it also follows from Theorem \ref{thm:unstable sn} below.} For completeness, we provide a direct sketch of proof.
\begin{proof}[Sketch of Proof]
Consider a population state in which a small fraction $\epsilon$ of the agents cooperate and the remaining $1-\epsilon$ agents defect. A revising agent most likely sees all the opponents defecting,  both in the $c$-sample and in the $d$-sample. With a probability of approximately $k \epsilon$, the agent sees a single cooperation in the $c$-sample and no cooperation in the $d$-sample, and so $c$ yields a mean payoff of $\frac{1-(k-1)\cdot l}{k}$ and $d$ yields $0$. The former is higher iff $l<\frac{1}{k-1}$. Thus, if the last inequality holds, then the prevalence of cooperation gradually  increases, and the population drifts away from the state where everyone defects. By contrast, if $l>\frac{1}{k-1}$, then  cooperation yields
the higher mean  payoff only if the $c$-sample includes at least two cooperators, which happens with a negligible probability of order $\epsilon ^2$. Therefore, in this case, cooperation gradually dies out, and the population converges to the state where everyone defects.
\end{proof}

\subsection{{Stability of (Partial) Cooperation}}\label{sec:main analysis}
Next we show that for any $k\geq2$,
{if $g$ and $l$ are  sufficiently small,
then} the prisoner's dilemma game admits a  globally asymptotically stable $S(k)$ equilibrium in which the frequency of cooperation is between $28\%$ and $50\%$ and is increasing in $k$.

\begin{theorem}\label{thm:mainresult}
{For} $k\geq2$ and $g,l<\frac{1}{k-1},$ the 
unique $S(k)$ globally asymptotically stable equilibrium $\alpha^k$  {in the prisoner's dilemma game} satisfies $0.28<\alpha^k_c <0.5$. Moreover, 
$\alpha^{k'}_c<\alpha^{k}_c$ for all $2 \le k'<k$.
\end{theorem}
{
The intuition for Theorem \ref{thm:mainresult} is as follows. The condition $g,l<\frac{1}{k-1}$ implies that cooperation yields a higher average payoff than defection iff the opponent has cooperated more times in the $d$-sample than in the $c$-sample. If $\alpha_c$ is close to zero, then the probability of this event is roughly $k \cdot \alpha_c>\alpha_c$. The symmetry between the two samples implies that the probability that the  opponent has cooperated more times in the $d$-sample is less than 0.5. Thus, there exists $\alpha^k_c<0.5$, which is not close to zero, for which the probability of cooperation yielding a higher average payoff is equal to $\alpha^k_c$. The  proof shows that this equality holds for $\alpha^k_c>0.28$, and that $\alpha_c$ is globally stable.
}
\begin{proof}
The proof of the theorem uses a number of claims, whose formal proofs are given in Appendix \ref{sec:proof-Thm1}. In what follows, we state each claim and present a sketch of proof. 

\noindent \textbf{Notation } For $j \leq k$, let $f_{k,p}(j)\equiv\binom{k}{j}p^j(1-p)^{k-j}$ 
be the probability mass function of a binomial random variable with parameters $k$ and $p.$ 
Let  $Tie(k,p)=\sum_{j=0}^k (f_{k,p}(j))^2$ be the probability of having a tie between two independent binomial random variables with parameters $k,p$, and let 
$Win(k,p)=0.5\cdot (1-Tie (k,p){)}$ be the probability that the first random variable has a larger value than the second. Let $p \equiv \alpha_c$ denote the proportion of cooperating agents in the population.

\begin{claim}\label{thm:BEPdynamic}\label{claim1}
Assume that $g,l\in (0,\frac{1}{k-1})$. The $k$-payoff sampling dynamic is given by
\begin{equation}\label{eq:BEPdyn}
    \dot{p} = Win(k,p)-p .
\end{equation}
\end{claim}
\begin{proof}[Sketch of Proof]
The condition $g,l<\frac{1}{k-1}$ implies that action $c$ has a higher mean payoff iff the $c$-sample includes more cooperating opponents than the $d$-sample does. The number of cooperators in each sample has a binomial distribution with parameters $k$ and $p$, and so the probability of
$c$ having a higher mean payoff  is $Win(k,p)$ (which we substitute in (\ref{eq:BEP})).
\end{proof}

For $k \geq 2$ and $0\leq p \leq 1$, denote the expression on the right-hand side of Eq. \eqref{eq:BEPdyn} by
\begin{equation}
h_{k}(p)=\label{eq:hn}Win(k,p)-p.
\end{equation}
\begin{claim}\label{thm:claim1}\label{claim2}
For $k \geq 2$, the function $h_k$ satisfies $h_k(0)=0, h_k(1) = -1$ and $h_k'(0) > 0.$
\end{claim}
\begin{proof}[Sketch of Proof]
When $p=0$ 
(resp., $=1$), in both samples all the opponents are defectors (resp.,  cooperators). The conclusion implies that $Win(k,p)=0$ for $p \in \{0,1\}$, which, in turn, implies that $h_k(0) = 0$ and $h_k(1) = -1$. Next, observe that for $p=\epsilon <<1$, $Win(k,p)\approx k\epsilon$, which is approximately the probability of having at least one cooperator in the $c$-sample. Thus, $h_k(\epsilon) \approx k\epsilon-\epsilon$, which implies that $h'_k(0)=k-1>0$.

\end{proof}
\begin{claim}\label{thm:claim2}\label{claim3}
For $k \geq 2$, the expression $h_k(p)$ is concave in $p$, and satisfies $h_k(p) < h_{k+1}(p)$ for $p \in (0,1)$, $h_k\left(\frac{1}{2}\right) < 0$, and $\lim_{k \rightarrow \infty}h_k\left(\frac{1}{2}\right) = 0.$
\end{claim}
\begin{proof}[Sketch of Proof]
 Observe that $Tie(k,p)$ is close to 1 when $p$ is close to either zero or one, and is smaller for intermediate $p$'s. The formal proof shows (by analyzing the characteristic function) that $Tie(k,p)$ is (1) convex in $p$, (2) decreasing in $k$ (i.e., the larger the number of actions in each sample, the smaller the probability of having exactly the same number of cooperators in both samples), and (3) converges to zero as $k$ tends to $\infty$. These findings imply that
  $h_k(p)=(0.5 \cdot (1-Tie(k,p))-p$ is concave in $p$ and increasing in $k$, and that 
  \[h_k\left(\frac{1}{2}\right) = \frac{1}{2}\cdot \left(1-Tie\left(k,\frac{1}{2}\right)\right)-\frac{1}{2} < \frac{1}{2}-\frac{1}{2} =0, \text{ ~~and}\] 
\[\lim_{k\rightarrow\infty}h_{k}\left(\frac{1}{2}\right)=\lim_{k\rightarrow\infty}\left(\frac{1}{2}\cdot \left(1-Tie\left(k,\frac{1}{2}\right)\right)-\frac{1}{2}\right)=\left(\frac{1}{2}-0\right)-\frac{1}{2}=0.\] 
\end{proof}

It follows from Claims \ref {claim2} and \ref{claim3} that for $k\geq2$ the equation $h_k(p)=0$ has a unique solution in the interval $(0,1)$, that this solution $p(k)$ corresponds to an $S(k)$ globally asymptotically stable state, that it satisfies $p(k)<0.5$, and that it is increasing in $k$. 

To complete the proof of \cref{thm:mainresult}, it remains to show that $p(2)>0.28$. This inequality is an immediate corollary of the  fact that for $p=0.28$,
\[h_2(p=0.28) = 2p(1-p)^3 + p^2(1-p^2) - p  \approx 0.001 > 0. ~~~~~\qedhere \]

 \end{proof}
Figure \ref{figure1} shows the $S(k)$ payoff sampling dynamics 
and the $S(k)$ globally stable equilibria for various values of $k$.

\begin{figure}[h]
\centering
\includegraphics[scale=0.8]{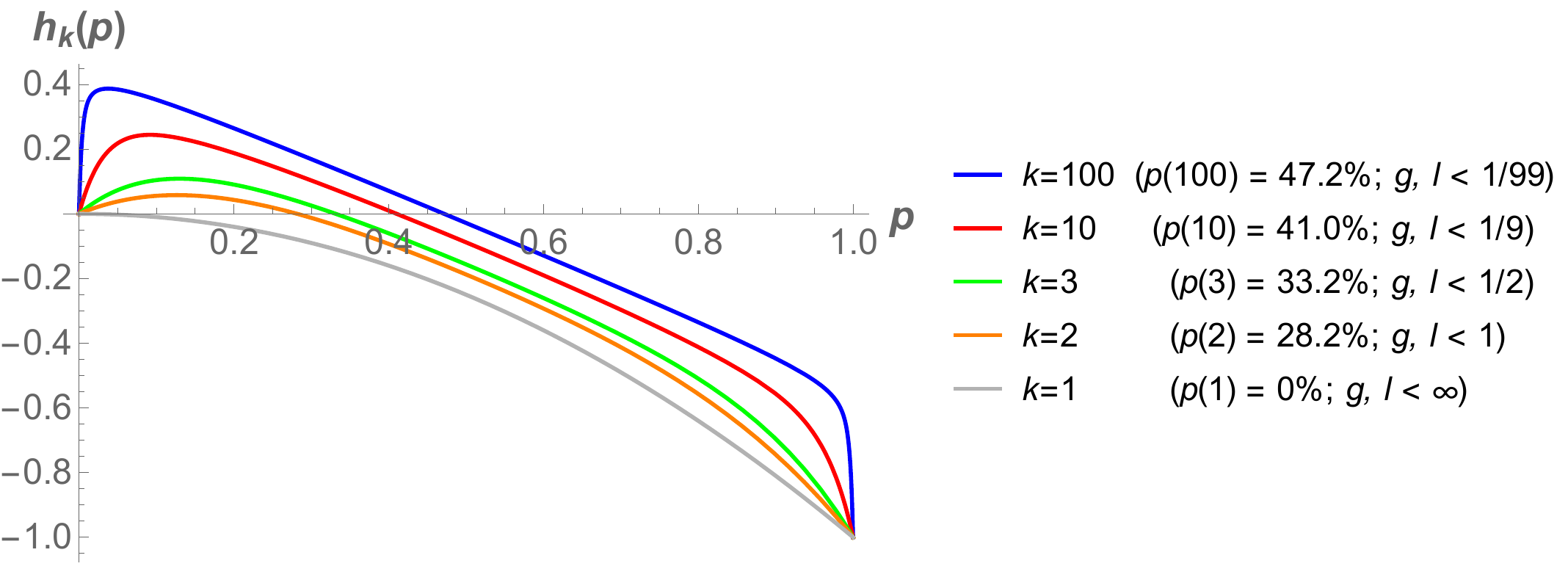}
\caption{The function $h_k$  
and its zero $p(k)$ for various values of $k$.}
\label{fig:BEP_PD}\label{figure1}
\end{figure}

\begin{remark}
{Theorem \ref{thm:mainresult} shows that partial cooperation is globally stable for any fixed $k$ if the parameters $g$ and $l$ are sufficiently small (with the upper bound depending on $k$). At the same time, it is well known that, for fixed $g$ and $l$, defection is globally stable if $k$ is sufficiently large (\citealp[Prop. 4.3]{sandholm2020stability}; see also   \citealp[Prop. 4]{osborne1998games}).}
\end{remark}
{
Theorem \ref{thm:mainresult} leaves open the question of the stability of partial cooperation when either $g$ or $l$ is larger than $\frac{1}{k-1}$. The next proposition answers this question for small sample sizes of $k=2$ or $3$, leaving full characterization of larger sample sizes for future research. As the proposition states, everyone defecting is globally asymptotically stable if $l>\frac{1}{k-1}$, while  if the reverse inequality holds, the globally asymptotically stable state has a substantial rate of cooperation of between $24\%$ and $33\%$.} 

\begin{proposition}\label{pro:k-2-3}
{For $k \in \{2,3\},$ the unique $S(k)$ globally asymptotically stable equilibrium $\alpha^k$ satisfies $\alpha^k_c=0$ if $l>\frac{1}{k-1}$, and $\alpha^k_c\in (0.24,0.33)$ if  $l<\frac{1}{k-1}$.}
\end{proposition}

The proof of the proposition is presented in Appendix \ref{sub-k-2-3}.

\section{{General Symmetric Games}}\label{sec:symmetric games}
 In this section we extend Proposition \ref{prop:Sk-strict-PD} by presenting a necessary and sufficient condition for $S(k)$ asymptotic stability of actions in general symmetric games. {Our characterization uses the following two definitions.}
 \subsection{Definitions}\label{subsec-defs}
\begin{definition}\label{def:generic}
A symmetric {$n$-player} game with payoff function $u : A^n \rightarrow \mathbb{R}$ is \textit{generic} if for any two sequences of action profiles $\left(\left(a_j^1, a_j^2,\dots,a_j^n\right)\right)_{j=1}^L$ and $\left(\left(\tilde{a}_j^1, \tilde{a}_j^2,\dots,\tilde{a}_j^n\right)\right)_{j=1}^L$ of equal length {$L$}, the equality
\[
\sum_{j=1}^L u\left(a_j^1, a_j^2,\dots,a_j^n\right) = \sum_{j=1}^L u\left(\tilde{a}_j^1, \tilde{a}_j^2,\dots,\tilde{a}_j^n\right)
\]
implies $\left\{a_j^1\right\}_{j=1}^L = \left\{\tilde{a}_j^1\right\}_{j=1}^L$. \end{definition}
 {Thus, a symmetric game is generic if the sums of payoffs of two sequences of action profiles are equal only if every action that appears as the first action in one of the profiles in one sequence also appears as the first action in the other sequence. Note that this definition of genericity is rather weak. A stronger definition could be used instead; the use of a weak  definition only strengthens our results below. Observe also that if each entry in the payoff matrix is independently drawn from a continuous (atomless) distribution, then the resulting random symmetric game is generic with probability one. Clearly, in a generic game, every pure Nash equilibrium is a strict equilibrium.}

\begin{definition}\label{def:support}
\emph{For an action $a^\ast$ in a  symmetric {$n$-player} game, and for $a,a' \in A\backslash\left\{ a^{*}\right\}$:}
\begin{enumerate}
\item 
\emph{action} $a$ directly $S(k)$ supports $a'$ against $a^\ast$ \emph{if}
{$$u(a', a,a^\ast,\dots,a^\ast) + (k-1)\cdot u(a', a^\ast,\dots,a^\ast) > k\cdot  u(a^\ast,\dots,a^\ast); \emph{ and}$$}
\item
\emph{action} $a$ $S(k)$ supports $a'$ by spoiling $a^\ast$ \emph{if}
$$ k\cdot u(a', a^\ast,\dots,a^\ast) > u(a^\ast, a,a^\ast,\dots,a^\ast)+(k-1) \cdot  u(a^\ast, \dots,a^\ast) $$
$$\text{~~\textrm{\emph{and}}~~}u(a',a^\ast,\dots,a^\ast)>u(b,a^\ast,\dots,a^\ast)~~\forall b\notin \{a^\ast,a'\}. $$
\end{enumerate}
{\emph{Action $a$ $S(k)$ \emph{single supports}, \emph{double supports,} or just \emph{supports} action $a'$ against $a^\ast$ if exactly one of conditions 1 and 2, both conditions, or at least one condition, respectively, holds. The notion of \emph{weak $S(k)$ support} and the related terms 
(weak direct support, weak support by spoiling, weak single support, and weak double support) are defined similarly, except that the strict inequalities in 1 and 2 are replaced by weak inequalities.}}
\end{definition}

{Less formally, action $a$ supports action $a'$ against action $a^\ast$ if a single appearance of $a$ in a population in which almost  everyone plays $a^\ast$ can make the mean payoff in the $a'$-sample the highest one.} For an action $a$ that directly supports $a'$ (a supporter of $a'$, {in the terminology of} \citeauthor{sandholm2020stability}, \citeyear[p.12]{sandholm2020stability}), a single appearance of $a$ in the $a'$-sample (with all other actions being $a^\ast$) is sufficient to make the mean payoff larger than that yielded by $a^\ast$, and thus to make it the largest payoff. For $a$ that supports $a'$ by spoiling $a^\ast$ (a benefiting spoiler of $a^\ast$, in that terminology), a single appearance of $a$ in the $a^\ast$-sample is sufficient to make the mean payoff smaller than that yielded by $a'$. This makes the latter the largest mean payoff {if} $a'$ is the second-best reply to $a^\ast$ 
(while if another action $a''$ is second best, then $a''$ yields a higher payoff, assuming that the $a''$-sample includes only $a^\ast$'s). 
Note that, in a generic game, the second-best reply is unique: the set $\argmax_{b\neq a^{*}}u(b,a^\ast,\dots,a^\ast)$ is a singleton. 

{Also, in a generic game, the notions of support and weak support coincide: action $a$ $S(k)$ supports $a'$ against $a^\ast$ if and only if it weakly $S(k)$ supports $a'$ against $a^\ast$.}

Observe that if $a^\ast$ is a symmetric equilibrium action, then $S(k)$ support against it is ``easier" the smaller $k$ is: {if action $a$ (weakly) $S(k)$ supports $a'$ against 
$a^\ast$, then it (respectively, weakly) $S(k')$ supports $a'$ against $a^\ast$ for all $k'<k$. }

\subsection{Result}
{It is well known (\citealp{osborne1998games}) that an action $a^\ast$ can be an $S(k)$ equilibrium only if $a^\ast$ is a symmetric Nash equilibrium (otherwise, $w_{a^\ast,k}(e_{a^\ast})=0$, which contradicts $e_{a^\ast}$ being a rest point).
Moreover, if the tie-breaking rule assigns positive probability to all co-winning actions, then $a^\ast$ must be a strict equilibrium (otherwise, $w_{a^\ast,k}(e_{a^\ast})<1$, which again contradicts $e_{a^\ast}$ being a rest point). 
Thus, being a strict symmetric equilibrium is essentially a necessary condition for an action to be $S(k)$ asymptotically stable. Our next result characterizes the conditions for a strict symmetric equilibrium action to be $S(k)$ asymptotically stable when
$k \geq 2$ or $n \geq 3$}.

\begin{theorem}\label{thm:unstable sn}
{Suppose that $k \geq 2$ or $n \geq 3$. A} necessary and sufficient condition for a strict symmetric equilibrium action $a^\ast$ in a generic symmetric {$n$-player} game to be $S(k)$ asymptotically stable is that, for the set $A^\ast \equiv A \backslash \{a^\ast\}$,
\begin{enumerate}
    \item
    [I.]
    {every nonempty subset $A' \subseteq A^\ast$ includes an action $a'$} that is not $S(k)$ supported  
    {against $a^\ast$} by any action in $A'$.
    \end{enumerate}
{
In a non-generic game, condition I is still  necessary  for $S(k)$ asymptotic stability, and a sufficient condition is}  
\begin{enumerate}
    \item
    [II.]
     every nonempty subset $A' \subseteq A^\ast$ includes an action $a'$ that is not  \em{weakly}
    \em{$S(k)$ supported against $a^\ast$ by any action in $A'$}.
    \end{enumerate}
{In conditions I and II, ``is not S(k) supported by'' and ``is not weakly S(k) supported by'' can be replaced by ``does not S(k) support'' and ``does not weakly S(k) support'', respectively, as the conditions resulting from these replacements, I' and II', are equivalent to I and II.}
\end{theorem}

\begin {proof}[Sketch of Proof]
Suppose that there is a nonempty subset of actions $A' \subseteq  A^\ast$, with cardinality 
$m\ge 1$, such that each action  $a' \in A'$ is supported by some action in $A'$. Consider an initial  population state $\alpha$ in which a fraction $1-\epsilon$ of the agents play $a^\ast$ and $\frac{\epsilon}{m}$ play each of the actions in $A'$. Since each $a' \in A'$ is supported by some action {$a$} in $A'$, {there is} a probability of approximately {$k(n-1) \frac{\epsilon}{m}$} of having action {$a$} appear in the sample and thus making the mean payoff yielded by $a'$ the highest one. It follows that {$w_{a',k}(\alpha) = k(n-1) \frac{\epsilon}{m} > \frac{\epsilon}{m} = \alpha_{a'}$} for all $a' \in A'$. Thus the frequency of all actions in $A'$ increases, 
{which implies that $a^\ast$ is not asymptotically stable}. The formal proof (Appendix \ref{sec:B3}) formalizes this intuition, by examining the Jacobian matrix at $e_{a^\ast}$ and showing that it admits an eigenvalue greater than 1.

Next, suppose 
that every nonempty subset $A' \subseteq A^\ast$ includes an action $a'$ that is not 
{weakly} $S(k)$ supported by any action in $A'$. Consider a state in which $1- \epsilon$ of the agents play $a^\ast$. We know that there exists an action $a'$ that is not {weakly} $S(k)$ supported by any action in $A^\ast$. This implies that the probability of action $a'$ having the maximal mean payoff in an agent's sample is $O(\epsilon ^2)$, and, thus, $w_{a', k}(\alpha)=O(\epsilon ^2)$. As the frequency of $a'$ becomes negligible, we can iterate the argument for $A' = A^\ast\backslash \{a'\}$, and find another action $a''$ for which  $w_{a'', k}(\alpha)=O(\epsilon ^2)$, etc. The formal proof (Appendix \ref{sec:B3}) shows that (a) condition II implies that all the eigenvalues of the Jacobian matrix are negative, and (b) the phrase ``is not {weakly} $S(k)$ supported by''  can be replaced by ``does not {weakly} $S(k)$ support.'' 
\end{proof}
{We remark that condition II is sufficient for} asymptotic stability also when $k=1$ {and $n=2$}. However, {condition I is not necessary in this case, as} demonstrated by defection in the prisoner's dilemma.

{
The following corollary shows that any  strict symmetric equilibrium is characterized by a threshold $k_0<  \infty$ that determines the equilibrium's asymptotic stability 
for any $k\ge 2$.
(Note that the corollary does not give any information regarding $S(1)$ stability.) }

\begin{corollary}
{Let $a^\ast$ be a strict symmetric equilibrium action in a symmetric game. There exists an integer $k_0$ such that, for $k \ge 2$, action $a^\ast$ is $S(k)$ asymptotically stable iff $k \ge k_0$.}
\end{corollary}
\begin {proof}
{
Let $\bar k \ge 2$ be a sufficiently large integer such that, for any action $a' \neq a^\ast$ and action profiles $(a^1,...,a^n)$ and $(a'^1,...,a'^n),$ $$\bar k\cdot\left(u(a^\ast,...,a^\ast)-u(a',a^\ast,...,a^\ast)\right)>u(a^1,...,a^n)-u(a'^1,...,a'^n).$$  
The inequality implies that there exists no pair of actions $a,a'\neq a^\ast$ such that action $a$ weakly supports action $a'$ against $a^\ast$, which in view of Theorem \ref{thm:unstable sn} implies that $a^\ast$ is an $S(\bar k)$ asymptotically stable equilibrium. 
Let $k_0$ be the smallest number for which $a^\ast$ is an $S(k_0)$ asymptotically stable equilibrium. 
This implies that there is 
no pair of actions $a,a'\neq a^\ast$ such that action $a$ $S(k_0)$ supports $a'$ against $a^\ast$. The last observation at the end of Section \ref{subsec-defs} implies that for any $k>k_0$ there is 
no pair of actions $a,a'\neq a^\ast$ such that action $a$ $S(k)$ supports $a'$ against $a^\ast$, which implies that, for $k\geq 2$, action $a^\ast$ is an $S(k)$ asymptotically stable equilibrium iff $k \geq k_0$.
}
\end{proof}

\subsection{Applications}\label{subsec-applications}
{In this subsection we demonstrate the usefulness of the necessary and sufficient conditions identified in Theorem  \ref{thm:unstable sn} by studying the $S(k)$ asymptotic stability of strict equilibria in a number of applications.} 

{
The first two applications deal with two-action games  ($A=\{a^\ast,a'\}$). Observe that in such games condition I (resp., I') in Theorem \ref{thm:unstable sn} reduces to the requirement that the other action $a'$ does not $S(k)$ (resp., weakly) support itself against $a^\ast$.}

\paragraph{\textbf{Prisoner's dilemma}} {Cooperation cannot support itself by spoiling, because it always increases the payoff of a defecting opponent. For $k \ge 2$, cooperation directly (resp., weakly) $S(k)$ supports itself against defection iff $l<\frac{1}{k-1}$ (resp., $\leq\frac{1}{k-1}$), because 
$$u(c,c)+(k-1)\cdot u(c,d)>k \cdot u(d,d)\Leftrightarrow~1+(k-1)\cdot(-l)>0~\Leftrightarrow~l<\frac{1}{k-1}.$$
By Theorem \ref{thm:unstable sn}, the last finding implies that defection is asymptotically stable if $l>\frac{1}{k-1}$ and is not asymptotically stable if $l<\frac{1}{k-1}$, which proves Proposition \ref{prop:Sk-strict-PD}.}

\paragraph{\textbf{Public good games}} {
Consider a symmetric $n$-player game, with $n \ge 2$, where the set of actions is 
$A=\{c,nc\}$, with $c$ and $nc$ interpreted as contributing or not contributing, respectively, a fixed amount of some private good.
The amount of public good produced is $\varphi(l)$, where $l$ is the number of contributions and $\varphi$ is a nondecreasing production function with $\varphi(0)=0$ and $\varphi(1) < 1$. Each of the contributors gets a payoff of $\varphi(l)-1$,
while for a non-contributor the payoff is $\varphi(l)$. 
The assumption  $\varphi(1) < 1$ implies that no one contributing is a strict equilibrium 
(and it is the unique equilibrium if  $\varphi(l+1)- \varphi(l)< 1$ for all $l$). 
Observe that contributing cannot support itself by spoiling, and it directly $S(k)$ supports itself against non-contributing iff 
$$u(c,c,nc,...,nc)+(k-1)\cdot u(c,nc,...,nc)>k\cdot   u(nc,...,nc)\Leftrightarrow$$
$$\varphi(2)-1+(k-1) \cdot (\varphi(1)-1)>0 \Leftrightarrow k<1-\frac{1-\varphi(2)}{1-\varphi(1)} .$$
By Theorem \ref{thm:unstable sn}, the above finding implies that 
if $n \ge 3$ or $k \ge 2$, then a sufficient condition for 
not contributing to be $S(k)$ asymptotically stable is that
$$k>1-\frac{1-\varphi(2)}{1-\varphi(1)} .$$
The corresponding weak inequality is a necessary condition.}

\paragraph {\textbf{Coordination games}} 
{Consider a symmetric $n$-player game, with $n \ge 2$, where the set of actions is $A=\{a_1,...,a_m\}$, with $m \ge 2$. If all players choose the same action $a_j$, everyone gets payoff $u_j,$ where $u_1\ge u_2\ge...\ge u_m>0.$ If all players do not choose the same action, then everyone gets zero. The game admits $m$ strict equilibria: everyone playing $a_1$, everyone playing $a_2$, ..., everyone playing $a_m$.}
{Next we characterize the stability of the strict symmetric equilibrium action $a_l$, for any $1\le m$. No action $a_i\neq a_l$ $S(k)$ supports an action $a_j \neq a_l$ against $a_l$ by spoiling, because $$u(a_j,a_l,...,a_l)=0<u(a_l,a_i,a_l,..,a_l)+(k-1)\cdot u(a_l,...,a_l)=(k-1)\cdot u_l.$$
The only action that might directly $S(k)$ support $a_j \neq a_l$ against $a_l$ is $a_j$ itself, and this can happen only if $n=2$. This is so because if $a_i \neq a_j$ or $n>2$, then the following inequality holds: $$u(a_j,a_i,a_l,...,a_l)+(k-1)\cdot u(a_j,a_l,...,a_l)=0<k\cdot u(a_l,...,a_l)=k\cdot u_l.$$ If $n=2$, then $a_j$ directly supports itself against $a_l$ iff $$u(a_j,a_j)+(k-1)\cdot u(a_j,a_l)>k\cdot u(a_l,a_l) \Leftrightarrow u_j+0>k \cdot u_l.$$ By Theorem \ref{thm:unstable sn}, this implies that all the strict equilibria are $S(k)$ asymptotically stable for any $k$ if there are at least three players. In the two-player case, and for $k \ge 2$, the strict symmetric equilibrium action $a_l$ is $S(k)$ asymptotically stable if $u_l>\frac{u_1}{k}$ and it is not $S(k)$ asymptotically stable if $u_l<\frac{u_1}{k}.$}
\subsection{Comparison with \citet{sandholm2020stability}}
We conclude this section by comparing Theorem \ref{thm:unstable sn} with the conditions for stability of strict equilibria
presented in \citet[Section 5]{sandholm2020stability} (which, in turn, improve on the conditions presented in \citealp{sethi2000stability}). {For simplicity, the comparison focuses on the case of generic symmetric games.}\footnote{
{In addition}, our setting  concerns only revising agents who test all feasible actions. The more general setting studied by \citeauthor{sandholm2020stability} allows dynamics in which revising agents test only some of the actions.} {As in Theorem \ref{thm:unstable sn}, we use the notation $A^\ast \equiv A \backslash \{a^\ast\}$.} 

\citet{sandholm2020stability} identify two necessary conditions for stability of a strict equilibrium. They are the negations of conditions 1 and 2 in the following proposition. 
\paragraph{\textbf{Adaptation of Proposition 5.4} \cite[]{sandholm2020stability}} 
\label{thm:sufficient}
Let $a^\ast$ be a strict symmetric equilibrium {action} 
in a {generic} symmetric game. {For} $k\geq 2$, action $a^\ast$ is \emph{not} $S(k)$ asymptotically stable if  either:
\begin{enumerate}
    \item  $\exists A' \subseteq A^\ast$ such that every $a'\in A'$ is directly supported by some action in $A'$; or
    \item  $\exists A' \subseteq A^\ast$ such that every action $a'\in A'$ supports some action in $A'$.
\end{enumerate}
Theorem \ref{thm:unstable sn} strengthens this result by omitting ``directly'' from condition 1, thus weakening the condition. Moreover, it shows that this weaker condition (call it 1') is actually equivalent to condition 2, {and that both 1' and 2 are in fact} necessary \emph{and sufficient} conditions for asymptotic stability. 

\cite{sandholm2020stability} present the following sufficient condition for stability.
\begin{definition}
{\emph{Action} $a$ tentatively $S(k)$ supports action $a'$} by spoiling $a^\ast$ \emph{if}
$${ k\cdot u(a', a^\ast,\dots,a^\ast) > u(a^\ast, a,a^\ast,\dots,a^\ast)+(k-1) \cdot  u(a^\ast,\dots,a^\ast).}$$
\end{definition}
\paragraph{\textbf{Adaptation of Prop. 5.9} \cite[]{sandholm2020stability}} 
Let $a^\ast$ be a strict symmetric equilibrium action in a generic symmetric game. {For} $k\geq 2$, action $a^\ast$ is $S(k)$ asymptotically stable  if 
\begin{enumerate}
    \item [3.] there exists an ordering of  $A^\ast$ such that no action $a$ {in this set} directly $S(k)$ supports or tentatively $S(k)$ supports by spoiling $a^\ast$ any weakly higher action $a'$.  
\end{enumerate}

Theorem \ref{thm:unstable sn} strengthens this result by omitting ``tentatively'' from condition 3, thus weakening the condition and making it \emph{necessary} and sufficient for stability.\footnote{
Condition 3 is not necessary for asymptotic stability. In the {symmetric two-player} game defined by the following payoff matrix, action $a^\ast$ is $S(2)$ asymptotically stable as it satisfies the condition in Theorem \ref{thm:unstable sn}, yet it does not satisfy condition 3 due to action $a''$ tentatively supporting itself by spoiling.\newline 
$ ~~~~~~~~~~~~~~~~~~~~~~~~~~~~~~~~~~~~~~~~~~~~~~~~~~~~~~~~~~~~~~~~~~~~~~~~~~~~~~$
\begin{tabular}{|c|ccc|}
\hline 
$a^{*}$ & 8 & 9 & 3\tabularnewline
$a$' & 7 & 5 & 2\tabularnewline
$a''$ & 6 & 4 & 1\tabularnewline
\hline 
\end{tabular}
}
Sufficiency still holds because the weaker condition (call it 3') implies that the lowest action in every subset $A' \subseteq A^\ast$ does not $S(k)$ support any action in $A'$. Necessity holds because it is not very difficult to see that 3' is implied by condition {I} in Theorem \ref{thm:unstable sn}. {(Order $A^\ast$ by recursively removing an element that is not $S(k)$ supported against $a^\ast$ by any of the current elements.)} 

\section{{Asymmetric Games}}\label{sec:asymmetric games}
{In what follows we adapt our model and the characterization of $S(k)$ asymptotic stability to asymmetric games. }

{Each player $i$ has a finite set of actions $A_i$ and a payoff function $u_i: \prod_{j=1}^n {A_j} \ \rightarrow \mathbb{R}$. The player is represented by a distinct population of agents, the $i$-population, whose state $\alpha^i$ is an element of the unit simplex in ${\mathbb{R}}^{|A_i|}$. The state of all $n$ populations is given by  $\alpha=(\alpha^1,\alpha^2,\dots,\alpha^n) \in \Delta$, where $\Delta$ is the Cartesian product of the players' unit simplices.} 

{The population state $\alpha$ determines for each player $i$ the probability vector $w^i_{k}(\alpha(t))$ specifying the probability that each of the player's actions yields the highest mean payoff in $k$ trials, employing some tie-breaking rule. For any $k \ge 1$, the $k$-payoff sampling dynamic is given by
\begin{equation}\label{eq:asymmetric BEP}
    \dot{\alpha}^i = w^i_{k}(\alpha(t)) - \alpha^i(t).
\end{equation}A population state $\alpha^\ast$ is an $S(k)$ equilibrium if $ w_k^i(\alpha^\ast) = (\alpha^{\ast})^i$ for each player $i$. Asymptotic stability and global asymptotic stability are defined as in the symmetric case.  
}

{
The notion of supporting an action 
against an action profile $a^\ast=(a^\ast_1,a^\ast_2,\dots,a^\ast_n)$ is conceptually similar to that in symmetric games. 
 To present it in a formally similar way, consider the disjoint union $A^\ast=\dot\bigcup_{i=1}^n{(A_i \backslash \{a^\ast_i\})}$. 
Each element of $A^\ast$ is of the form $a_i$: a specific action of a specified player $i$ such that $a_i\ne a_i^\ast$. For such an element, $(a_i,a^{\ast}_{-i})$ denotes the action profile in which player $i$ plays $a_i$ and all the other players play according to $a^\ast$. For $a_i,a_j \in A^\ast$ with $i \neq j$, $(a_i,a_j,a^{\ast}_{-ij})$ denotes the action profile in which player $i$ plays $a_i$, player $j$ plays $a_j$, and all the other players play according to $a^*$.}

\begin{definition}\label{def:asymmetric_support}
{
\emph{For an action profile $a^\ast$ in an $n$-player game, and for $a_i,a_j \in A^\ast$:}
\begin{enumerate}
\item 
\emph{action} $a_i$ directly $S(k)$ supports $a_j$ against $a^\ast$ \emph{if} $i \ne j$ \emph{and}
$$u_j(a_i,a_j,a^{\ast}_{-ij}) + (k-1)\cdot u_j(a_j,a^{\ast}_{-j}) > k\cdot  u_j(a^\ast);\emph{ and}$$
\item
\emph{action} $a_i$ $S(k)$ supports $a_j$ by spoiling $a^\ast$ \emph{if} $i \ne j$ \emph{and}
$$ k\cdot u_j(a_j,a^{\ast}_{-j}) > u_j(a_i,a^{\ast}_{-i})+(k-1) \cdot  u_j(a^\ast) \text{~~~\emph{and}~~~}u_j(a_j,a^
\ast_{-j})>u_j(b_j,a^\ast_{-j})~~\forall b_j\neq a_j\in A^\ast.$$
\end{enumerate}
\emph{Action} $a_i$ $S(k)$ single supports, double supports, \emph{or just} supports \emph{action $a_j$ against $a^\ast$ if exactly one of conditions 1 and 2, both conditions, or at least one condition, respectively, holds.}
Weak \emph{$S(k)$ support, and the related terms, are defined similarly, except that the strict inequalities in 1 and 2 are replaced by weak inequalities.}}
\end{definition}
Next we adapt the characterization of $S(k)$ asymptotic stability to asymmetric games.
\begin{theorem}\label{thm:asymmetric unstable sn}
{For $k \ge 2$, a necessary condition for a strict equilibrium $a^\ast$ in an $n$-player game to be $S(k)$ asymptotically stable is that condition I (equivalently, I') in Theorem \ref{thm:unstable sn} holds, and a sufficient condition is that II (equivalently, II') holds.}
\end{theorem}

{Obviously, the necessary conditions coincide with the sufficient ones if the game  is generic, in the standard sense. The proof of the theorem, which is very similar to that of Theorem \ref{thm:unstable sn}, is presented in Appendix \ref{sec:asymmetric proof}. }

{When the underlying game is symmetric, there are two different best experienced payoff dynamics that are applicable to it. The baseline, \emph{one-population} dynamics presented in Section \ref{sec:model} (specifically, (\ref{eq:BEP})) assumes a single population from which the players are sampled. Moreover, players are not assigned roles in the game; there is no player 1, player 2, etc. An alternative dynamics that can be applied to the game are the \emph{$n$-population} dynamics \ref{eq:asymmetric BEP}. Although meant for asymmetric games, they can be used to study symmetric games in which the players are arbitrarily numbered, with the $i$-population representing player\footnote{{The $n$-population dynamics can also capture environments  in which players from a single population play in all roles, the roles are observable, and a player conditions her action on her role.}}} $i$.

{In other evolutionary dynamics it is often the case that stability under the one-population dynamics is not equivalent to stability under the $n$-population dynamics. (For example, it is well known that the mixed equilibrium of a hawk-dove game is stable under the one-population replicator dynamics but is not stable under the two-population replicator dynamics.) Our next result shows that this is not the case here.}

\begin{corollary}\label{cor-stric-tsymmetric}
{For $ k \ge 2$, a strict symmetric equilibrium action $a^\ast$ in a symmetric $n$-player game is $S(k)$ asymptotically stable under the one-population dynamics (\ref{eq:BEP}) if and only if everyone playing $a^\ast$ is $S(k)$ asymptotically stable under the $n$-population dynamics (\ref{eq:asymmetric BEP}).}
\end{corollary}

{The simple proof, which is given in Appendix \ref{subsec:cor-proof}, relies on the fact that our two definitions of an action supporting another action against a strict symmetric equilibrium (action) essentially coincide when the underlying game is symmetric.}

\subsection{Applications}

{We conclude this section with demonstrating the usefulness of Theorem \ref{thm:asymmetric unstable sn} by applying it to the study of $S(k)$ asymptotic stability of strict equilibria in asymmetric prisoner's dilemma and hawk-dove games.}
\begin{table}[h]
\centering %
\begin{tabular}{c|cc}
\multicolumn{3}{c}{Asymmetric Prisoner's Dilemma}\tabularnewline
 & \textcolor{red}{\emph{$c_{2}$ }} & \textcolor{red}{\emph{$d_{2}$}}\tabularnewline
\hline 
\textcolor{blue}{\emph{$c_{1}$}} & \textcolor{blue}{1}~~,~~\textcolor{red}{1} & \textcolor{blue}{~~-$l_{1}$}~,~\textcolor{red}{1+$g_{2}$}\tabularnewline
\textcolor{blue}{\emph{$d_{1}$}} & \textcolor{blue}{1+$g_{1}$}~,~\textcolor{red}{-$l_{2}$~~~} & \textcolor{blue}{0}~~,~~\textcolor{red}{0}\tabularnewline
\end{tabular}~~~~~~~~~~~~~~~%
\begin{tabular}{c|cc}
\multicolumn{3}{c}{Asymmetric Hawk-Dove}\tabularnewline
 & \textcolor{red}{\emph{$D_{2}$}} & \textcolor{red}{\emph{$H_{2}$}}\tabularnewline
\hline 
\textcolor{blue}{\emph{$D_{1}$}} & \textcolor{blue}{1}~~,~~\textcolor{red}{1} & \textcolor{blue}{~~$l_{1}$}~,~\textcolor{red}{1+$g_{2}$}\tabularnewline
\textcolor{blue}{\emph{$H_{1}$}} & \textcolor{blue}{1+$g_{1}$}~,~\textcolor{red}{$l_{2}$~~~} & \textcolor{blue}{0}~~,~~\textcolor{red}{0}\tabularnewline
\end{tabular}
\centering{}\caption{Payoff Matrices of Asymmetric Games (${\color{blue}g_{1},g_{2}},{\color{red}l_{1},l_{2}}>0$; 
in hawk-dove, also $\color{blue}l_{1},\color{red}l_{2}<1)\label{tab:asymmetric-games}$}
\end{table}

\paragraph{\textbf{Asymmetric prisoner's dilemma}}
{The left-hand side of Table \ref {tab:asymmetric-games} presents the payoff matrix of an asymmetric prisoner's dilemma, in which the unique equilibrium is $d=(d_1,d_2)$, mutual defection. 
Action $c_1$ cannot support $c_2$ by spoiling, because cooperation increases the payoff of a defecting opponent. It directly (weakly) supports $c_2$ against $d$ iff $l_2 < \frac{1}{k-1}$ (respectively,  $\le \frac{1}{k-1}$), because
$$u_2(c_1,c_2)+(k-1)\cdot u_2(d_1,c_2)>k \cdot u_2(d)\Leftrightarrow~1+(k-1)\cdot(-l_2)>0.$$
The same holds with 1 and 2 interchanged. This implies, by Theorem \ref{thm:asymmetric unstable sn}, that for any $k \ge 2$ mutual defection is $S(k)$ asymptotically stable if $\max(l_1,l_2)>\frac{1}{k-1}$ and is not $S(k)$ asymptotically stable if $\max(l_1,l_2)<\frac{1}{k-1}.$ 
Note that, if $l_1=l_2$, then in accordance with Corollary \ref{cor-stric-tsymmetric} the condition for stability coincides with that of the one-population model in the symmetric prisoner's dilemma.
}
\paragraph{\textbf{Asymmetric hawk-dove}}
{The right-hand side of Table \ref {tab:asymmetric-games} presents the payoff matrix of the asymmetric hawk-dove game. The game admits two strict equilibria, $(D_1,H_2)$ and $(H_1,D_2)$, in which one player plays hawk and the other plays dove. 
Observe that action $D_2$ can $S(k)$ support action $H_1$ against $(D_1,H_2)$ only by direct support, which is obtained iff
\[
u_{1}\left(H_{1},D_{2}\right)+\left(k-1\right)\cdot u_{1}\left(H_{1},H_{2}\right)>k\cdot u_{1}\left(D_{1},H_{2}\right)\,\Leftrightarrow\,1+g_{1}>k\cdot l_{1}.
\]
Action $H_1$ can $S(k)$ support action $D_2$ against $(D_1,H_2)$ only  by spoiling, which is obtained iff
\[
k\cdot u_{2}\left(D_{1},D_{2}\right)>u_{2}\left(H_{1},H_{2}\right)+\left(k-1\right)\cdot u_{2}\left(D_{1},H_{2}\right)\,\Leftrightarrow\,k>\left(k-1\right)\cdot\left(1+g_{2}\right)\,\Leftrightarrow\,g_{2}<\frac{1}{k-1}.
\]
By Theorem \ref{thm:asymmetric unstable sn}, this implies that $(D_1,H_2)$ is asymptotically stable if 
$g_{1}<k\cdot l_{1}-1$ or $g_{2}>\frac{1}{k-1}$ and is not asymptotically stable if $g_{1}>k\cdot l_{1}-1$ and $g_{2}<\frac{1}{k-1}$.
The conditions for asymptotic stability of the other strict equilibrium are obtained by interchanging 1 and 2. 
}

\addcontentsline{toc}{section}{\protect\numberline{}Appendix}

\renewcommand{\thesection}{\Alph{section}}
\setcounter{section}{0}
\section*{Appendix}
\section{Proofs}\label{sec:proof-Thm1}

\subsection{Proof of  \cref{thm:BEPdynamic}}
Recall that $p \equiv \alpha_c$ and $1-p \equiv \alpha_d$ denote the proportion of agents in the population playing actions $c$ and $d$, respectively. An agent's $c$-sample includes $k$ actions of the opponents. If $j \in \{0,...,k\}$ of these are $c$, then the agent's mean payoff (when playing $c$) is $j-(k-j)l$. 
Similarly, if $j' \in \{0, ... ,k\}$ of the sampled  actions in the agent's $d$-sample are $c$, then the mean payoff (when playing action $d$) is $j'(1+g)$. 
The difference between the two payoffs is
\[
(j-(k-j)l)-j'(1+g)=j-j'-((k-j)l+j'g).
\]This expression is clearly negative if $j \le j'$. If $j \ge j'+1$, it is positive, since $(k-j)l+j'g \le \frac{k-j+j'}{k-1} \le 1$ and the first inequality becomes an equality only if $j=k$ and $j'=0$ while the second one does so only if $j = j'+1$.  

These conclusions prove that the $c$-sample yields a superior payoff iff it includes more cooperations than the $d$-sample does. As the number of cooperators in each sample has a binomial distribution with parameters $k$ and $p$, we conclude that  $w_{c,k} = Win(k,p)$.

\subsection{Proof of \cref{thm:claim1}}\label{sec:A1}
A binomial random variable with parameters $k,p$ has a degenerate distribution if $p\in \{0,1\}$, and so $Win(k,0)=Win(k,1)=0$. From Eq. \eqref{eq:hn} we get $h_{k}(0) =0-0=0, h_{k}(1)=0-1=-1$.
To show that $h_k'(0) > 0$ for $k = 2,3,4, \dots$, we use the fact that

\begin{align*}
    h_k(p) & = Win(k,p)-p = 0.5(1-Tie(k,p))-p=0.5(1-\sum_{j=0}^k (f_{k,p}(j))^2)-p)\\
    & = 0.5(1-((1-p)^{2k}+O(p^2))-p  ~~~~~~~(O(p^2)\  \text{denotes the terms with degree} \geq 2)\\
    & =  0.5(1-(1-2pk+O(p^2))-p + O(p^2) = kp - p + O(p^2) = (k-1)p + O(p^2).
    \end{align*}
From the above expression, it follows that $h_k'(0) = k-1>0$ for $k>1$.

\subsection{Proof of \cref{thm:claim2}}\label{sec:A2}
Let $\{A_{j,p}\}_{j=1}^k,$ $\{B_{j,p}\}_{j=1}^k$ be $2k$ independent Bernoulli random variables with parameter $p.$ Then $X_{k,p} = \sum_{i=1}^k A_{j,p}$ and $Y_{k,p} = \sum_{j=1}^k B_{j,p}$ are i.i.d. binomial random variables with parameters $k$ and $p.$ Eq. \eqref{eq:hn} can be expressed as 
\begin{align}\label{eq:555}
    h_k(p) & = P(X_{k,p} > Y_{k,p}) - p \nonumber \\
    & = \frac{1}{2}\left(P(X_{k,p} > Y_{k,p}) + P(X_{k,p} < Y_{k,p})\right) - p \hspace{0.1in} (\text{since} \ X_{k,p} \ \text{and} \ Y_{k,p} \ \text{are} \ i.i.d. ) \nonumber \\
    & = \frac{1}{2}\left(1-P(X_{k,p} = Y_{k,p})\right) - p.  
\end{align}

\noindent For $j = 1,2,\dots,k,$ let $Z_{j,p} = A_{j,p} - B_{j,p}.$ Clearly, $\{Z_{j,p}\}_{j=1}^n$ are i.i.d., with distribution given by
\[
P(Z_{j,p} = -1) = P(Z_{j,p} = 1) = pq \ \text{ and } \ P(Z_{j,p} = 0) = 1-2pq,  
\]
where $q=1-p.$ 

Consider the characteristic function $\varphi_k(\cdot;p)$ of the random variable $Z_p^k = X_{k,p}-Y_{k,p} = \sum_{j=1}^k Z_{j,p}:$ 

\begin{align*}
\varphi_k(t;p) & = \mathbb{E}\left[e^{itZ_p^k}\right]  = \mathbb{E}\left[e^{it\sum_{j=1}^k Z_{j,p}}\right] = \left(\mathbb{E}[e^{itZ_{1,p}}]\right)^k \\
 & = \left(e^{it(-1)}pq + e^{it(1)}pq + e^{it(0)}(1-2pq)\right)^k = \left(1+pq(e^{-it}+e^{it}-2)\right)^k \\
  & = \left(1+2pq(\text{cos}\ (t)-1)\right)^k \hspace{0.45in} (\text{since} \ e^{-it}+e^{it} = 2\text{cos}\ (t)) \\
   & = \left(1-4p(1-p)\text{sin}^2\left(\frac{t}{2}\right)\right)^k \hspace{0.3in} \left(\text{since} \ q=1-p \ \text{and} \ \text{cos}\ (t) = 1-2\text{sin}^2\left(\frac{t}{2}\right)\right).
\end{align*}
The base of the last exponent, with power $k$, is an expression that is convex as a function of $p$, lies between 0 and 1 and, if $0 < p < 1$ and $t$ is not a whole multiple of $\pi$, is different from $0$ and $1$. Therefore, the same is true for $\varphi_k(t;p)$ and, if $0 < p <1$ and $t$ is not a whole multiple of $\pi,$ $\varphi_k(t;p) > \varphi_{k+1}(t;p)$ and $\lim_{k \rightarrow \infty}\varphi_k(t;p) = 0.$ It follows, in view of Eq. \eqref{eq:555} and the fact that (see \cref{thm:fact1}) 
\[
P(X_{k,p} = Y_{k,p}) = P(Z_p^k = 0) = \frac{1}{2\pi}\int_{-\pi}^{\pi}\varphi_k(t;p)dt, 
\]
that the function $h_k$ is concave and, for $0<p<1,$ the sequence $\left(h_k(p)\right)_{k=1}^{\infty}$ is strictly increasing and converges to $\frac{1}{2} - p.$ This completes the proof.

\begin{fact}
\label{thm:fact1}
$P(Z_p^k = 0) = \frac{1}{2\pi}\int_{-\pi}^{\pi}\varphi_k(t;p)dt.$
\end{fact}
\begin{factproof} From the definition of $\varphi_k(t;p),$ we have the following:
\begin{align*}
    \int_{-\pi}^{\pi}\varphi_k(t;p)dt & = \int_{-\pi}^{\pi}\mathbb{E}\left[e^{itZ_p^k}\right]dt = \mathbb{E}\left[\int_{-\pi}^{\pi}e^{itZ_p^k}dt \right] = \mathbb{E}\left[\int_{-\pi}^{\pi}e^{itZ_p^k}\left(\mathds{1}_{\{Z_p^k = 0\}}+ \mathds{1}_{\{Z_p^k \neq 0\}}\right)dt\right] \\
    & = \mathbb{E}\left[\int_{-\pi}^{\pi}e^{itZ_p^k}\mathds{1}_{\{Z_p^k = 0\}}dt\right] + \mathbb{E}\left[\int_{-\pi}^{\pi}e^{itZ_p^k}\mathds{1}_{\{Z_p^k \neq 0\}}dt\right] \\
     & = \mathbb{E}\left[\int_{-\pi}^{\pi}1 \cdot \mathds{1}_{\{Z_p^k = 0\}}dt\right] + \mathbb{E}\left[0\right] = \mathbb{E}\left[2\pi \cdot \mathds{1}_{\{Z_p^k = 0\}}\right] = 2\pi\cdot P(Z_p^k = 0).
\end{align*}
From the above series of equalities, it follows that $P(Z_p^k = 0) = \frac{1}{2\pi}\int_{-\pi}^{\pi}\varphi_k(t;p)dt.$
\end{factproof}

\subsection{Proof of \cref{thm:unstable sn}}\label{sec:B3}
For completeness,  we present all details of the proof, although various steps are analogous to arguments presented in the proofs of \citet[Section 5]{sandholm2020stability}.

Suppose, first, that the game is generic. Let $T$ be the nonnegative $|A^\ast| \times |A^\ast|$ matrix  whose element in row $a'$ and column $a$ is
\[
 T_{a' a} =
  \begin{cases}
   2 \ \text{if action } a \text{ double } S(k) \text{ supports action } a' \text{ against } a^\ast \\
   1 \ \text{if action } a \text{ single } S(k) \text{ supports action } a' \text{ against } a^\ast \\
   0 \ \text{if action } a \text{ does not } S(k) \text{ support action } a' \text{ against } a^\ast \\
  \end{cases}
\]
We will now compute the Jacobian of the $k$-payoff sampling dynamic \eqref{eq:BEP} in the (monomorphic population state corresponding to the) strict symmetric equilibrium action $a^{*}.$ Suppose that the frequency $\alpha_{a^{*}}$ of action $a^{*}$ in the population is $1-\epsilon$, where $\epsilon>0$ is a small number. Denote by $\alpha^{*} \equiv \alpha|_{A^{*}}$ the frequencies of the actions in the set $A^{*}$, that is, $\alpha^{*}_{a} = \alpha_{a}$ for all $a \in A^{*}.$ Clearly, $\sum_{a \in A^{*}}\alpha^{*}_{a} = \epsilon$, which implies that the Euclidean norm of $\alpha^{*}$ is of order $\epsilon$, that is, $|\alpha^{*}| = \left(\sum_{a \in A^{*}}\alpha^{*2}_{a}\right)^{\frac{1}{2}} = O(\epsilon).$

The probability that a non-equilibrium action $a' \in A^\ast$ yields the best payoff  is roughly equal to the probability that it yields a higher payoff than the strict symmetric equilibrium action $a^\ast$ does when both actions are tested $k$ times in the population state $\alpha$. When $a' \in A^\ast$ is tested $k$ times, with a very high probability (of $(1-\epsilon)^{k(n-1)}$) it encounters the equilibrium
action $a^\ast$ each time. The probability that $a^\ast$ is encountered $k-2$ or fewer times is of order $\epsilon^2$ or higher, and these higher-order terms can be neglected for stability analysis. Action $a'$ yields a higher payoff than $a^\ast$ does in the following two cases:

Case 1: When $a'$ is tested, one of the $k(n-1)$ opponents plays a non-equilibrium action $a$ and the remaining opponents play $a^\ast.$ Action $a'$ obtains the highest mean payoff if $a$ directly $S(k)$ supports it against $a^\ast$, that is, $u(a', a,a^\ast,\dots,a^\ast) + (k-1)\cdot u(a', a^\ast,\dots,a^\ast) > k\cdot  u(a^\ast, a^\ast,\dots,a^\ast)$.

Case 2: When $a^\ast$ is tested, one of the $k(n-1)$ opponents plays $a'$ and the others play $a^\ast.$  Action $a'$ obtains the maximal mean payoff if  $a$ $S(k)$ supports it by spoiling $a^\ast$, that is, $ k\cdot u(a', a^\ast,\dots,a^\ast) > u(a^\ast, a,a^\ast,\dots,a^\ast)+(k-1) \cdot  u(a^\ast,\dots,a^\ast)$ and $u(a', a^\ast,\dots,a^\ast)>u(b,a^\ast,\dots,a^\ast)$ for all $b \notin \{a', a^\ast\}$.

The probability that action $a'$ yields the best payoff is therefore given by
\begin{align}\label{eq:BEP555}
    w_{a',k}(\alpha) & = k(n-1)\sum_{a \in A^\ast}T_{a' a}\alpha_{a} + O(|\alpha^\ast|^2),
\end{align}
and the $k$-payoff sampling dynamic \eqref{eq:BEP} can be written as 
\begin{align*}
    \dot{\alpha}_{a'} & = w_{a', k}(\alpha) - \alpha_{a'} = k(n-1)\sum_{a \in A^\ast}T_{a' a}\alpha_{a} - \alpha_{a'}+ O(|\alpha^\ast|^2) .
\end{align*}
In matrix notation, 
\begin{equation}\label{eq:nfreq}
    \dot{\alpha}^\ast = f(\alpha^\ast) \equiv (k(n-1)T-I)\alpha^\ast + O(|\alpha^\ast|^2),
\end{equation}
where $I$ is the $|A^\ast| \times |A^\ast|$ identity matrix and $O(|\alpha^\ast|^2)$ here is an $|A^\ast|$-dimensional vector with elements of order $|\alpha^\ast|^2$ or higher. Let $J$ denote the Jacobian matrix of $f$ evaluated at the origin:
\begin{equation*}
    J  = \left.\frac{\partial f(\alpha^\ast)}{\partial \alpha^\ast}\right|_{\alpha^\ast = \underbrace{(0, 0, \dots , 0)}_{|A^\ast| \ \text{zeros}}} = k(n-1)T-I.
\end{equation*}
The asymptotic stability of the system \eqref{eq:nfreq} can be analyzed by examining the eigenvalues of the Jacobian matrix $J.$ 

A sufficient condition for $a^\ast$ to be $S(k)$ asymptotically stable is that all the eigenvalues of $J$ have negative real parts
(see, e.g., \citealp[Corollary 8.C.2]{sandholm2010population}). A sufficient condition for it \emph{not} to be $S(k)$ asymptotically stable is that at least one of the eigenvalues has a positive real part. The first condition holds, in particular, if the only eigenvalue of $T$ is zero, in other words, if the spectral radius $\rho$ of that matrix is $0$, as this condition means that the only eigenvalue of $J$ is $-1$. The second condition holds if $\rho\ge 1$. This is because the spectral radius of a nonnegative matrix is an eigenvalue with a nonnegative eigenvector (\citealp{johnson1985matrix}, Theorem 8.3.1), and so $\rho\ge 1$ implies that $J$ has the eigenvalue $k(n-1)\rho -1\ge 2\cdot 1-1>0$, with a corresponding nonnegative eigenvector. It therefore suffices to show that $\rho=0$ holds if condition I in the theorem holds and $\rho\ge 1$ holds if the condition does not hold. {Observe that conditions I and I' can be rephrased as follows: every principal submatrix of $T$ has a row or column, respectively, where all entries are zero. Therefore, to complete the proof of the theorem, it remains only to establish the following fact (which in particular proves the equivalence of I and I').}  

\begin{fact}\label{thm:fact3}
Let $M$ be a square matrix of nonnegative integers, and let $\rho$ be its spectral radius. If every principal submatrix of $M$ has a row of zeros, then $\rho=0$. Otherwise, $\rho\ge 1$. The same is true when we replace ``row'' by ``column.''
\end{fact}
\begin{factproof} Suppose that  $\rho>0$. Let $v$ be a corresponding nonnegative right eigenvector. Since $(Mv)_i = \rho v_i > 0$ for every index $i$ with $v_i > 0$, the set $\alpha$ of all such indices defines a principal submatrix $M'$ (obtained from $M$ by deleting all rows and all columns with indices not in $\alpha$) with the property that every row includes at least one nonzero entry.

Conversely, suppose that $M$ has a principal submatrix $M'$, defined by some set of indices $\alpha$, with the above property. Let $v$ be a column vector of $0$'s and $1$'s where an entry is $1$ if and only if its index lies in $\alpha$. It is easy to see that $Mv \ge v$. This vector inequality implies that  $\rho\ge 1$ (\citealp{johnson1985matrix}, Theorem 8.3.2). 

The proof for ``column'' is obtained by replacing $M$ by its transpose. 
\end{factproof}

{We now drop the assumption that the game is generic. This change means that cases 1 and 2 need to be extended by replacing the respective strict inequalities with weak ones. When a weak inequality in case 1 or 2 holds as equality, the probability of moving to action $a'$ depends on the tie-breaking rule. Adding these probabilities to $T_{a' a}$, for all $a',a \in A^\ast$, defines a new matrix $\bar{T} \ge T$, which replaces $T$ in \eqref{eq:nfreq}. If the probabilities were replaced by $1$'s, the result would be a matrix $\bar{\bar{T}} \ge \bar{T}$, where all entries are integers. In view of Fact \ref{thm:fact3}, conditions II and II' are both equivalent to the condition that the spectral radius of $\bar{\bar{T}}$ is $0$. That condition implies that the spectral radius of $\bar{T}$ is also $0$, which implies that $a^\ast$ is $S(k)$ asymptotically stable.} If condition I or I' does not hold, then the spectral radius of $T$ is $1$ or greater. The same is then true for $\bar{T}$, which implies that $a^\ast$ is not $S(k)$ asymptotically stable.      

\subsection{Proof of \cref{thm:asymmetric unstable sn}}\label{sec:asymmetric proof}
{Let $T$ be the $|A^\ast| \times |A^\ast|$ matrix defined exactly as in the proof of \cref{thm:unstable sn}, and denote by $\alpha^{*} \equiv \alpha|_{A^{*}}$ the vector of dimension $\sum_{i=1}^n{(|A_i|-1)}$ whose components are the frequencies of the actions in the set $A^{*}$ (that is, the players' non-equilibrium actions). By arguments similar to those employed in that proof, if the game is generic, then the Jacobian $J$ of the $k$-payoff sampling dynamic \eqref{eq:asymmetric BEP} is given by $J=kT-I$, so that
$$    \dot{\alpha}^\ast = (kT-I)\alpha^\ast + O(|\alpha^\ast|^2).$$
The rest of the proof, including the treatment of the non-generic case, is essentially the same as for \cref{thm:unstable sn}. }

\subsection{{Proof of Proposition \ref{pro:k-2-3}: Global Stability for $k \in \{2,3\}$}}\label{sub-k-2-3}

In this subsection, we characterize the $S(2)$ {and $S(3)$ globally asymptotically stable equilibria} explicitly for all parameter configurations $g, l$ {in the prisoner's dilemma}.

{$S(2)$ analysis:}
In Theorem \ref{thm:mainresult} we solved the case $g, l < 1.$ Here, we consider the remaining three cases.
Recall that, when testing each action twice, we have:
\[
 \text{
 When a player samples }\  c \text{ she}
  \begin{cases}
   \text{gets} \ 2 \ \text{with probability} \ p^2 \\
   \text{gets} \ 1-l \ \text{with probability} \ 2p(1-p) \\
   \text{gets} \ -2l \ \text{with probability} \ (1-p)^2
  \end{cases}
\]
\[
 \text{When a player samples}\ d \text{ she}
  \begin{cases}
   \text{gets} \ 2(1+g) \ \text{with probability} \ p^2 \\
   \text{gets} \ 1+g \ \text{with probability} \ 2p(1-p) \\
   \text{gets} \ 0 \ \text{with probability} \ (1-p)^2
  \end{cases}
\]

\textbf{Case I:} $l < 1 < g.$ 
Action $c$ has a higher mean payoff iff the $c$-sample includes at least one cooperation and the $d$-sample does not include any cooperation. Thus, the 
$2$-\textit{payoff sampling dynamic} in this case is given by
\begin{align*}
    \dot{p}& = p^2(1-p)^2 + 2p(1-p)(1-p)^2 - p = (1-p)^2(p^2 + 2p(1-p)) - p\\
        & = (1-p)^2(1-(1-p)^2) - p = (1-p)^2 - (1-p)^4 - p.
\end{align*}
The rest points of the above dynamic are $0$ and $0.245.$ It is straightforward to verify that $0$ is unstable and that $0.245$ is globally stable.

\textbf{Case II:} $g < 1 < l.$ 
Action $c$ has a higher mean payoff iff the $c$-sample includes two cooperations, while the $d$-sample includes at most one cooperation. The $2$-\textit{payoff sampling dynamic} in this case is given by $\dot{p} = p^2(1-p^2) - p$. The unique rest point is $0,$ which is globally stable.

\textbf{Case III:} $g, l > 1.$ 
Action $c$ has a higher mean payoff iff the $c$-sample includes two cooperations, while the $d$-sample does not include any cooperation.
The $2$-\textit{payoff sampling dynamic} in this case is given by $\dot{p} = p^2(1-p)^2 - p$. As in the previous case, the unique rest point is $0,$ which is globally stable.

{$S(3)$ analysis:} {In Theorem \ref{thm:mainresult} we solved the case $g, l < \frac{1}{2}.$ Here, we consider the remaining cases.
Recall that, when testing each action thrice, we have:}

\[
 \text{When a player samples}\  c \text{ she}
  \begin{cases}
   \text{gets} \ 3 \ \text{with probability} \ p^3 \\
   \text{gets} \ 2-l \ \text{with probability} \ 3p^2(1-p) \\
   \text{gets} \ 1-2l \ \text{with probability} \ 3p(1-p)^2 \\
   \text{gets} \ -2l \ \text{with probability} \ (1-p)^3
  \end{cases}
\]
\[
 \text{When a player samples}\ d \text{ she}
  \begin{cases}
   \text{gets} \ 3(1+g) \ \text{with probability} \ p^3 \\
   \text{gets} \ 2(1+g) \ \text{with probability} \ 3p^2(1-p) \\
   \text{gets} \ 1+g \ \text{with probability} \ 3p(1-p)^2 \\
   \text{gets} \ 0 \ \text{with probability} \ (1-p)^3
  \end{cases}
\]

{
\textbf{Case I:} $l < \frac{1}{2}, \frac{1}{2} < g < 2$, and $g+l < 1.$
Action $c$ has a higher mean payoff iff, when the $c$-sample includes at least two cooperations, the $d$-sample includes at most one cooperation, or when the $c$-sample includes exactly one cooperation, the $d$-sample does not include any cooperation. Thus, the $3$-\textit{payoff sampling dynamic} in this case is given by
\begin{align*}
\dot{p} &= p^3(3p(1-p)^2+(1-p)^3) + 3p^2(1-p)(3p(1-p)^2+(1-p)^3) + 3p(1-p)^2 (1-p)^3 - p\\
& = p^2(1-p)^2(3-2p)(1+2p) + 3p(1-p)^5 - p.
\end{align*}
The rest points of the above dynamic are $0$ and $0.323.$ It is straightforward to verify that $0$ is unstable and that $0.323$ is globally stable.}

{
\textbf{Case II:} $l < \frac{1}{2}, \frac{1}{2} < g < 2$, and $g+l > 1.$
Action $c$ has a higher mean payoff iff, when the $c$-sample includes three cooperations, the $d$-sample includes at most one cooperation, or when the $c$-sample includes either one or two cooperations, the $d$-sample does not include any cooperation. Thus, the $3$-\textit{payoff sampling dynamic} in this case is given by
\begin{align*}
\dot{p} & = p^3(3p(1-p)^2+(1-p)^3) + 3p^2(1-p)(1-p)^3 + 3p(1-p)^2 (1-p)^3 - p \\
 & = p^3(1-p)^2(1+2p) +3p(1-p)^4 - p. 
\end{align*}
The rest points of the above dynamic are $0$ and $0.250.$ It is straightforward to verify that $0$ is unstable and that $0.250$ is globally stable.}

{
\textbf{Case III:} $l < \frac{1}{2}, g > 2.$
Action $c$ has a higher mean payoff iff, when the $c$-sample includes at least one cooperation, the $d$-sample does not include any cooperation. Thus, the $3$-\textit{payoff sampling dynamic} in this case is given by
\begin{align*}
\dot{p} & = p^3(1-p)^3 + 3p^2(1-p)(1-p)^3 + 3p(1-p)^2 (1-p)^3 - p \\
 & = (1-(1-p)^3)(1-p)^3 - p = (1-p)^3 - (1-p)^6 - p.
\end{align*}
The rest points of the above dynamic are $0$ and $0.245.$ It is straightforward to verify that $0$ is unstable and that $0.245$ is globally stable.}

{
\textbf{Case IV:} $\frac{1}{2} < l < 2, g < \frac{1}{2}$ and $g+l<1.$
Action $c$ has a higher mean payoff iff, when the $c$-sample includes three cooperations, the $d$-sample includes at most two cooperations, or when the $c$-sample includes exactly two cooperations, the $d$-sample includes at most one cooperation. Thus, the $3$-\textit{payoff sampling dynamic} in this case is given by
\begin{align*}
\dot{p} & = p^3(1-p^3) + 3p^2(1-p)(3p(1-p)^2)+(1-p)^3) - p \\
 & = p^3(1-p^3) + 3p^2(1-p)^3(1+2p) - p.
\end{align*}
The unique rest point is $0,$ which is globally stable.}

{
\textbf{Case V:} $\frac{1}{2} < l < 2, g < \frac{1}{2}$, and $g+l>1.$
Action $c$ has a higher mean payoff iff, when the $c$-sample includes three cooperations, the $d$-sample includes at most two cooperations, or when the $c$-sample includes exactly two cooperations, the $d$-sample does not include any cooperation. Thus, the $3$-\textit{payoff sampling dynamic} in this case is given by
\begin{align*}
\dot{p} & = p^3(1-p^3) + 3p^2(1-p)(1-p)^3  - p = p^3(1-p^3) + 3p^2(1-p)^4  - p.
\end{align*}
The unique rest point is $0,$ which is globally stable.}

{
\textbf{Case VI:} $\frac{1}{2} < l < 2, \frac{1}{2} < g < 2.$
Action $c$ has a higher mean payoff if, when the $c$-sample includes three cooperations, the $d$-sample includes at most one cooperation, or when the $c$-sample includes exactly two cooperations, the $d$-sample does not include any cooperation. Thus, the $3$-\textit{payoff sampling dynamic} in this case is given by
\begin{align*}
\dot{p} & = p^3(3p(1-p)^2+(1-p)^3) + 3p^2(1-p)(1-p)^3 - p \\
& = p^3(1-p)^2(1+2p) +3p^2(1-p)^4 - p. 
\end{align*}
The unique rest point is $0,$ which is globally stable.}

{
\textbf{Case VII:} $\frac{1}{2} < l < 2, g > 2.$
Action $c$ has a higher mean payoff iff, when the $c$-sample includes at least two cooperations, the $d$-sample does not include any cooperation. Thus, the $3$-\textit{payoff sampling dynamic} in this case is given by
\begin{align*}
\dot{p} & = p^3(1-p)^3 + 3p^2(1-p)(1-p)^3 - p = p^2(1-p)^3(3-2p) - p. 
\end{align*}
The unique rest point is $0,$ which is globally stable.}

{
\textbf{Case VIII:} $l > 2, g < \frac{1}{2}.$
Action $c$ has a higher mean payoff iff, when the $c$-sample includes three cooperations, the $d$-sample includes at most two cooperations. Thus, the $3$-\textit{payoff sampling dynamic} in this case is given by
$\dot{p} = p^3(1-p^3) - p.$ The unique rest point is $0,$ which is globally stable.}

{
\textbf{Case IX:} $l > 2, \frac{1}{2} < g < 2.$
Action $c$ has a higher mean payoff iff, when the $c$-sample includes three cooperations, the $d$-sample includes at most one cooperation. Thus, the $3$-\textit{payoff sampling dynamic} in this case is given by
\begin{align*}
\dot{p} & = p^3(3p(1-p)^2+(1-p)^3) - p  = p^3(1-p)^2(1+2p) - p. 
\end{align*}
The unique rest point is $0,$ which is globally stable.}

{
\textbf{Case X:} $l > 2, g > 2.$
Action $c$ has a higher mean payoff iff, when the $c$-sample includes three cooperations, the $d$-sample does not include any cooperation. Thus, the $3$-\textit{payoff sampling dynamic} in this case is given by $\dot{p} = p^3(1-p)^3 - p.$ The unique rest point is $0,$ which is globally stable.}

\subsection{Proof of Corollary \ref{cor-stric-tsymmetric}}\label{subsec:cor-proof}
{
Suppose that action $a^\ast$ is not $S(k)$ asymptotically stable under the one-population dynamics, and so there is a subset $A' \subseteq A\backslash \{a^\ast\}$ such that all actions in $A'$ are supported against $a^\ast$ by actions in $A'$. Let $\bar A'=\dot\bigcup_{i=1}^n{A'}$ be the disjoint union of $n$ copies of $A'.$ It follows immediately from Definitions \ref{def:support} and \ref{def:asymmetric_support} and the symmetry of the game that all actions in $\bar A'$ are supported by actions in $\bar A'$ against $\bar{a}^\ast \equiv (a^\ast,a^\ast,\dots,a^\ast)$. Therefore, the strategy profile $\bar{a}^\ast$ is not $S(k)$ asymptotically stable under the $n$-population dynamics.}

{Conversely, suppose that the last conclusion holds, so that there is a subset $A' \subseteq A^\ast \equiv  \dot\bigcup_{i=1}^n{(A \backslash \{a^\ast\})}$ such that all actions in $A'$ are supported by actions in $A'$ against $\bar{a}^\ast$. Let $\bar A'=\{a\in A \mid \text{there is some player } i \text{ and a corresponding action } a_i \in A'\text{ with }a=a_i\}$ be the set of all actions that are included in $A'$ for at least one player. It is easy to see that all actions in $\bar A'$ are supported against $a^\ast$ by actions in $\bar A'$, and so action $a^\ast$ is not $S(k)$ asymptotically stable under the one-population dynamics.}

\section{{Values of $g$ and $l$ in Prisoner's Dilemma Experiments}}\label{sec-experiments-g-l}
{
Figure \ref{figure_gl} shows the values of $g$ and $l$ in the 29 experiments of the one-shot prisoner's dilemma (taken from 16 papers) as summarized in the meta-study of \citet[Table A.3]{mengel2018risk}. The figure shows that most of these experiments satisfy the condition for global stability of partial cooperation for $k=2$ (namely, $l<1$), and 
quite a few of them  also satisfy the condition for $k=3$ ($l<0.5$).}

\begin{figure}

\caption{{Values of $g$ and $l$ in the 29 Experiments Summarized in \citet[Table A.3]{mengel2018risk} \label{figure_gl}}}

\includegraphics[scale=0.6]{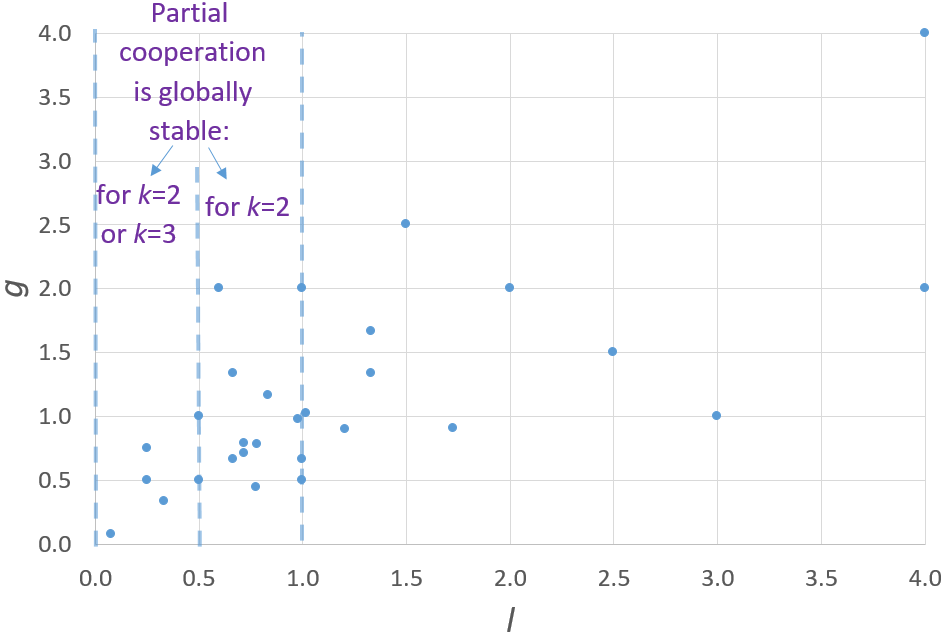}
\end{figure}
\newpage
\mybibliography

\begin{thebibliography}{}

\bibitem[Burton-Chellew et~al., 2017]{burton2017social}
Burton-Chellew, M.~N., El~Mouden, C., and West, S.~A. (2017).
\newblock Social learning and the demise of costly cooperation in humans.
\newblock {\em Proceedings of the Royal Society B: Biological Sciences},
  284\text{:20170067}.

\bibitem[C{\'a}rdenas et~al., 2015]{cardenas2015stable}
C{\'a}rdenas, J., Mantilla, C., and Sethi, R. (2015).
\newblock Stable sampling equilibrium in common pool resource games.
\newblock {\em Games}, 6(3):299--317.

\bibitem[Chmura and G{\"u}th, 2011]{chmura2011minority}
Chmura, T. and G{\"u}th, W. (2011).
\newblock The minority of three-game: An experimental and theoretical analysis.
\newblock {\em Games}, 2(3):333--354.

\bibitem[Fehr and Schmidt, 1999]{fehr1999theory}
Fehr, E. and Schmidt, K.~M. (1999).
\newblock A theory of fairness, competition, and cooperation.
\newblock {\em The Quarterly Journal of Economics}, 114(3):817--868.

\bibitem[Heller and Mohlin, 2018]{heller2018social}
Heller, Y. and Mohlin, E. (2018).
\newblock Social learning and the shadow of the past.
\newblock {\em Journal of Economic Theory}, 177:426--460.

\bibitem[Horn and Johnson, 1985]{johnson1985matrix}
Horn, R.~A. and Johnson, C.~R. (1985).
\newblock {\em Matrix Analysis}.
\newblock Cambridge University Press.

\bibitem[Kosfeld et~al., 2002]{kosfeld2002myopic}
Kosfeld, M., Droste, E., and Voorneveld, M. (2002).
\newblock A myopic adjustment process leading to best-reply matching.
\newblock {\em Games and Economic Behavior}, 40(2):270--298.

\bibitem[Kreindler and Young, 2013]{kreindler2013fast}
Kreindler, G.~E. and Young, H.~P. (2013).
\newblock Fast convergence in evolutionary equilibrium selection.
\newblock {\em Games and Economic Behavior}, 80:39--67.

\bibitem[Mantilla et~al., 2018]{mantilla2018efficiency}
Mantilla, C., Sethi, R., and C{\'a}rdenas, J.~C. (2018).
\newblock Efficiency and stability of sampling equilibrium in public goods
  games.
\newblock {\em Journal of Public Economic Theory}, 22(2):355--370.

\bibitem[McKelvey and Palfrey, 1995]{mckelvey1995quantal}
McKelvey, R.~D. and Palfrey, T.~R. (1995).
\newblock Quantal response equilibria for normal form games.
\newblock {\em Games and economic behavior}, 10(1):6--38.

\bibitem[Mengel, 2018]{mengel2018risk}
Mengel, F. (2018).
\newblock Risk and temptation: A meta-study on prisoner's dilemma games.
\newblock {\em The Economic Journal}, 128(616):3182--3209.

\bibitem[Mi{\k{e}}kisz and Ramsza, 2013]{mikekisz2013sampling}
Mi{\k{e}}kisz, J. and Ramsza, M. (2013).
\newblock Sampling dynamics of a symmetric ultimatum game.
\newblock {\em Dynamic Games and Applications}, 3(3):374--386.

\bibitem[Nax et~al., 2016]{nax2016learning}
Nax, H.~H., Burton-Chellew, M.~N., West, S.~A., and Young, H.~P. (2016).
\newblock Learning in a black box.
\newblock {\em Journal of Economic Behavior \& Organization}, 127:1--15.

\bibitem[Nax and Perc, 2015]{nax2015directional}
Nax, H.~H. and Perc, M. (2015).
\newblock Directional learning and the provisioning of public goods.
\newblock {\em Scientific Reports}, 5(1):1--6.

\bibitem[Osborne and Rubinstein, 1998]{osborne1998games}
Osborne, M.~J. and Rubinstein, A. (1998).
\newblock Games with procedurally rational players.
\newblock {\em American Economic Review}, 88(4):834--847.

\bibitem[Oyama et~al., 2015]{oyama2015sampling}
Oyama, D., Sandholm, W.~H., and Tercieux, O. (2015).
\newblock Sampling best response dynamics and deterministic equilibrium
  selection.
\newblock {\em Theoretical Economics}, 10(1):243--281.

\bibitem[Rabin, 1993]{rabin1993incorporating}
Rabin, M. (1993).
\newblock Incorporating fairness into game theory and economics.
\newblock {\em American Economic Review}, 83(5):1281--1302.

\bibitem[Ramsza, 2005]{ramsza2005stability}
Ramsza, M. (2005).
\newblock Stability of pure strategy sampling equilibria.
\newblock {\em International Journal of Game Theory}, 33(4):515--521.

\bibitem[Rowthorn and Sethi, 2008]{rowthorn2008procedural}
Rowthorn, R. and Sethi, R. (2008).
\newblock Procedural rationality and equilibrium trust.
\newblock {\em The Economic Journal}, 118(530):889--905.

\bibitem[Rustichini, 2003]{rustichini2003equilibria}
Rustichini, A. (2003).
\newblock Equilibria in large games with continuous procedures.
\newblock {\em Journal of Economic Theory}, 111(2):151--171.

\bibitem[Salant and Cherry, 2020]{salant2020statistical}
Salant, Y. and Cherry, J. (2020).
\newblock Statistical inference in games.
\newblock {\em Econometrica}, 88(4):1725--1752.

\bibitem[Sandholm, 2001]{sandholm2001almost}
Sandholm, W.~H. (2001).
\newblock Almost global convergence to p-dominant equilibrium.
\newblock {\em International Journal of Game Theory}, 30(1):107--116.

\bibitem[Sandholm, 2010]{sandholm2010population}
Sandholm, W.~H. (2010).
\newblock {\em Population Games and Evolutionary Dynamics}.
\newblock MIT Press.

\bibitem[Sandholm et~al., 2019]{sandholm2019best}
Sandholm, W.~H., Izquierdo, S.~S., and Izquierdo, L.~R. (2019).
\newblock Best experienced payoff dynamics and cooperation in the centipede
  game.
\newblock {\em Theoretical Economics}, 14(4):1347--1385.

\bibitem[Sandholm et~al., 2020]{sandholm2020stability}
Sandholm, W.~H., Izquierdo, S.~S., and Izquierdo, L.~R. (2020).
\newblock Stability for best experienced payoff dynamics.
\newblock {\em Journal of Economic Theory}, 185:104957.

\bibitem[Sethi, 2000]{sethi2000stability}
Sethi, R. (2000).
\newblock Stability of equilibria in games with procedurally rational players.
\newblock {\em Games and Economic Behavior}, 32(1):85--104.

\bibitem[Sethi, 2019]{Raj}
Sethi, R. (2019).
\newblock Procedural rationality in repeated games.
\newblock Unpublished manuscript, Barnard College, Columbia University.

\bibitem[Spiegler, 2006a]{spiegler2006competition}
Spiegler, R. (2006a).
\newblock Competition over agents with boundedly rational expectations.
\newblock {\em Theoretical Economics}, 1(2):207--231.

\bibitem[Spiegler, 2006b]{spiegler2006market}
Spiegler, R. (2006b).
\newblock The market for quacks.
\newblock {\em The Review of Economic Studies}, 73(4):1113--1131.

\end{thebibliography}
\end{document}